\documentclass[a4paper,USenglish,cleveref,autoref]{lipics-v2021} 
\pdfoutput=1 
\hideLIPIcs  

\graphicspath{{figures/}}
\usepackage{amssymb, amsfonts}
\usepackage{doi} 
\usepackage{bold-extra} 
\usepackage{float}
\usepackage{verbatim}
\usepackage{mathtools}
\usepackage{tikz}
\usepackage{layouts}
\newtoggle{fullversion}
\toggletrue{fullversion} 
\newcommand{\N}{{\textup{\textsc{n}}}}
\renewcommand{\S}{{\textup{\textsc{s}}}}
\newcommand{\E}{{\textup{\textsc{e}}}}
\newcommand{\W}{{\textup{\textsc{w}}}}
\newcommand{\NE}{{\textup{\textsc{ne}}}}
\newcommand{\SW}{{\textup{\textsc{sw}}}}
\newcommand{\NW}{{\textup{\textsc{nw}}}}
\newcommand{\SE}{{\textup{\textsc{se}}}}

\newcommand{\hole}{H}
\newcommand{\phole}{\tilde{H}}
\newcommand{\outline}{O}
\newcommand{\rhombus}{\mathcal{R}}
\renewcommand{\O}{\mathcal{O}}

\title{Distributed Rhombus Formation of Sliding Squares}

\author{Irina Kostitsyna}{KBR at NASA Ames Research Center, USA}{irina.kostitsyna@nasa.gov}{https://orcid.org/0000-0003-0544-2257}{}
\author{David Liedtke}{Paderborn University, Germany}{liedtke@mail.upb.de}{https://orcid.org/0000-0002-4066-0033}{}
\author{Christian Scheideler}{Paderborn University, Germany}{scheideler@upb.de}{https://orcid.org/0000-0002-5278-528X}{}

\authorrunning{I. Kostitsyna, D. Liedtke, C. Scheideler} 
\Copyright{Irina Kostitsyna and David Liedtke and Christian Scheideler}



\category{} 

\relatedversion{} 

\funding{This work was supported by the DFG Project SCHE 1592/10-1.}

\nolinenumbers 

\EventEditors{John Q. Open and Joan R. Access}
\EventNoEds{2}
\EventLongTitle{42nd Conference on Very Important Topics (CVIT 2016)}
\EventShortTitle{CVIT 2016}
\EventAcronym{CVIT}
\EventYear{2016}
\EventDate{December 24--27, 2016}
\EventLocation{Little Whinging, United Kingdom}
\EventLogo{}
\SeriesVolume{42}
\ArticleNo{23}

\begin{document}
    
\maketitle              

\begin{abstract}
    The sliding square model is a widely used abstraction for studying self-reconfigurable robotic systems, where modules are square-shaped robots that move by sliding or rotating over one another. 
    In this paper, we propose a novel distributed algorithm that allows a group of modules to reconfigure into a rhombus shape, starting from an arbitrary side-connected configuration. 
    It is connectivity-preserving and operates under minimal assumptions: one leader module, common chirality, constant memory per module, and visibility and communication restricted to immediate neighbors.
    Unlike prior work, which relaxes the original sliding square move-set, our approach uses the unmodified move-set, addressing the additional challenge of handling locked configurations. 
    Our algorithm is sequential in nature and operates with a worst-case time complexity of $\O(n^2)$ rounds, which is optimal for sequential algorithms.
    To improve runtime, we introduce two parallel variants of the algorithm. 
    Both rely on a spanning tree data structure, allowing modules to make decisions based on local connectivity. 
    Our experimental results show a significant speedup for the first variant, and linear average runtime for the second variant, which is worst-case optimal for parallel algorithms.
\end{abstract}


\section{Introduction}

Edge-connected configurations of square modules, which can reconfigure by sliding over or rotating around one another (see \cref{subfig:def-moves}), are a well-established theoretical model for self-reconfigurable robotic systems.
There are numerous applications such as search-and-rescue missions or autonomous construction, where modules collectively navigate complex terrains or form shapes while preserving connectivity.
Reconfiguration planning in the sliding square model has primarily focused on centralized approaches.
In contrast, we propose a fully distributed algorithm that enables modules to autonomously self-organize into a rhombus shape under minimal assumptions about their initial positions, computational capabilities, and sensory range.
The rhombus serves as an efficient intermediate shape due to its compactness, minimal bounding box, and lack of holes.

\subsection{Our Contribution}

We introduce a distributed reconfiguration algorithm for the sliding square model that transforms any initially side-connected configuration with a designated leader and common chirality into a rhombus shape.
These two assumptions are minimal but necessary to break initial symmetries.
The algorithm has a worst-case runtime of $\mathcal{O}(n^2)$ rounds, which is optimal among solutions that build up any shape sequentially.
It operates under strict local communication and visibility constraints, requires only constant memory per module, and maintains connectivity throughout execution.
In contrast to previous distributed solutions, which relax the original sliding square move-set, our algorithm utilizes the unmodified move-set, introducing the additional challenge of handling locked configurations.
While our main focus is the rhombus, the algorithm generalizes to any shape that can be sequentially constructed, provided that at every step, any two cells in the shape remain connected by a shortest path within the shape. 

To improve efficiency, we introduce two parallel variants of the algorithm.
Both variants rely on an auxiliary spanning tree structure maintained throughout execution, allowing modules to detect certain moves that preserve connectivity.
The first variant permits modules to move only if such moves improve their depth relative to the spanning tree (without explicitly computing or storing depth), resulting in a significant speedup over the sequential version. 
However, some configurations remain locked, limiting parallelization and requiring sequential progression.
The second variant allows modules to move whenever possible, showing an average-case linear runtime, as demonstrated by our experiments.

\subsection{Related Work}

Reconfiguration in the sliding square model, first introduced by Fitch et al. \cite{Fitch2003-ReconfigurationPlanning}, typically involves transforming configurations into a canonical intermediate shape before reaching a desired target shape.
Dumitrescu and Pach \cite{Dumitrescu2004-PushingSquaresAround} were the first to provide a universal 2D reconfiguration algorithm in this model, achieving a worst-case optimal runtime of $\mathcal{O}(n^2)$ moves using a straight line as the canonical shape.
Moreno and Sacristán \cite{Moreno2020-ReconfiguringSlidingSquares} adapted this into an in-place version, restricting all intermediate states to lie within the 1-offset bounding box of the initial and target configurations, while still requiring $\Theta(n^2)$ moves.
Akitaya et al. \cite{Akitaya2021-CompactingSquares} presented the first input sensitive algorithm, achieving $\mathcal{O}(nP)$ complexity, where $P$ is the maximum perimeter of the bounding boxes, using $xy$-monotone shapes and proving that minimizing the number of moves is NP-hard.
Abel et al. \cite{Abel2024-UniversalInPlaceReconfiguration} give the first universal reconfiguration algorithm for sliding cubes in three and higher dimensions, achieving an optimal $\mathcal{O}(n^2)$ bound and an in-place variant where all but one module stay within the bounding boxes of the start and target configurations. 
Kostitsyna et al. \cite{Kostitsyna2023-OptimalIn-PlaceCompactionOfSlidingCubes} achieve asymptotically optimal in-place reconfiguration for sliding squares, with the number of moves proportional to the total coordinate sum over all cubes.

All aforementioned algorithms are centralized and sequential. 
The first distributed and synchronous approach was by Dumitrescu et al. \cite{Dumitrescu2004-MotionPlanningForMetamorphicSystems}, who reconfigure $y$-monotone shapes into a line, assuming global communication.
Hurtado et al. \cite{Hurtado2015-DistributedReconfiguration} later gave the first distributed universal algorithm with optimal $\mathcal{O}(n)$ rounds, but with a relaxed move set, a leader requiring $\mathcal{O}(n)$ memory, and $\mathcal{O}(\log n)$ memory in other modules (despite their claim of constant memory, as storing integers up to $n$ requires $\log n$ bits). 
They further assume a vision and communication range of $2$. 
Wolters \cite{Wolters2024-ParallelAlgorithms} presented parallel versions for reconfiguring $xy$-monotone shapes, though their work lacks correctness proofs. 
Recently, Akitaya et al. \cite{Akitaya2024-SlidingSquaresInParallel} gave the first centralized theoretical results for parallel reconfiguration in the sliding square model, achieving worst-case optimal makespan $\mathcal{O}(P)$ and showing NP-completeness for constant-makespan decisions.

In the limited sliding cube model, modules are restricted to sliding moves only, and reconfiguration is achieved through meta-modules, small groups of modules that act as single units. 
Kawano proposed a series of increasingly efficient algorithms in this model, culminating in a linear-time solution \cite{Kawano2015-CompleteReconfiguration,Kawano2020-Distributed,Kawano2023-Linear-TimeReconfiguration}.

Pivoting models, which allow only corner rotations, are significantly harder. 
Akitaya et al. \cite{Akitaya2020-CharacterizingUniversalReconfigurability, Akitaya2023-Reconfiguration} demonstrated PSPACE-hardness for several pivoting reconfiguration problems in both 2D and 3D.
A notable workaround uses five auxiliary modules (called musketeers) to achieve universal reconfiguration in $\mathcal{O}(n^2)$ pivot moves \cite{Akitaya2021-UniversalReconfigurationOfFacet-ConnectedModularRobotsByPivots}.
Non-universal strategies focus on configurations excluding certain forbidden substructures. 
Sung et al. \cite{Sung2015-ReconfigurationPlanning} identified three such subconfigurations, allowing reconfiguration in $\mathcal{O}(n^2)$ moves, which was extended by Feshbach and Sung \cite{Feshbach2021-ReconfiguringNon-ConvexHoles} to include non-convex holes. 

Other relevant problems include coordinated motion planning, where multiple agents navigate a shared space while maintaining connectivity \cite{Fekete2021-ConnectedCoordinatedMotionPlanningWithBoundedStretch, Walter2005-AlgorithmsForFastConcurrentReconfiguration}, nest-building inspired by biological systems to form compact structures \cite{Czyzowicz2021-BuildingANestByAnAutomaton}, and bounding-box construction \cite{Fekete2020-RecognitionAndReconfigurationOfLatticeBasedCellularStructures}.

\iftoggle{fullversion}{
    Amoebot-based models offer a theoretical framework for programmable matter in which simple, memory-constrained agents self-organize via expansion and contraction on a grid without global coordination. 
    There are various publications demonstrating that these systems are capable of complex shape formation tasks \cite{Derakhshandeh2015-AmoebotShapeFormation, Derakhshandeh2016-AmoebotShapeFormation, DiLuna2020-AmoebotShapeFormation}.
    A closely related variant is the hybrid programmable matter model, in which only a small subset of mobile agents is active while the majority of the system consists of passive tiles that are manipulated by the active ones. 
    These systems are capable of efficient shape formation in both 2D and 3D settings \cite{Gmyr2020-FormingTileShapesWithSimpleRobots, Hinnenthal2020-Hybrid3D, Hinnenthal2024-ShapeFormation, Friemel2025-ShapeReconfiguration}.
}{
    Amoebot-based models offer a theoretical framework for programmable matter, where memory-constrained agents locally self-organize and reconfigure into various complex shapes~\cite{Derakhshandeh2015-AmoebotShapeFormation, Derakhshandeh2016-AmoebotShapeFormation, DiLuna2020-AmoebotShapeFormation}. 
    A related variant is hybrid programmable matter, where few active agents manipulate passive tiles, enabling efficient self-reconfiguration in both 2D and 3D settings~\cite{Gmyr2020-FormingTileShapesWithSimpleRobots, Hinnenthal2020-Hybrid3D, Hinnenthal2024-ShapeFormation, Friemel2025-ShapeReconfiguration}.
}

These theoretical results are complemented by a growing body of real-world implementations in modular robotics, such as \cite{Jenett2019-MaterialRobotSystemforAssemblyOfDiscreteCellularStructure, Gregg2024-Ultralight}.

\subsection{Model Definition and Problem Statement}

Consider the infinite square grid with set of cells $C$.
Two distinct cells $c_1, c_2 \in C$ are called \emph{adjacent} if they share at least one corner, \emph{side-adjacent} if they share an edge, and \emph{corner-adjacent} if they share a corner but no edge.
The \emph{neighborhood} $N(c)$ of a cell $c$ consists of its eight adjacent cells: four side-adjacent cells in directions north ($\N$), east ($\E$), south ($\S$), and west ($\W$), and four corner-adjacent cells in directions northeast ($\NE$), southeast ($\SE$), southwest ($\SW$), and northwest ($\NW$). 

We consider a set of $n$ square modules, each with the computational power of a deterministic finite automaton.  
Each module occupies precisely one cell on the grid, and each cell contains at most one module.  
For simplicity, we refer interchangeably to modules and the cells they occupy.
A \emph{configuration} $(M, s, l)$ consists of the set of all modules $M \subseteq C$, a state function $s : M \rightarrow \mathcal{S}$ assigning each module $m \in M$ a state $s(m) \in \mathcal{S}$ from the finite set $\mathcal{S}$ of possible states, and a unique cell $l \in C$ initially occupied by the leader module.

Modules move according to two allowed operations: \emph{slides} and \emph{convex transitions} (see \cref{subfig:def-moves}).  
Let $(m, a, b, c)$ form a 4-cycle of side-adjacent cells, where $m \in M$.  
If $a \notin M$ and $b, c \in M$, then module $m$ can \emph{slide} into the empty cell $a$.  
If $a, b \notin M$ and $c \in M$, then $m$ can perform a \emph{convex transition} into the empty cell $b$.  

We assume a fully synchronous activation schedule consisting of discrete rounds.  
In each round, every module simultaneously executes a \emph{Look-Compute-Move} cycle:
In the \emph{look} phase, a module $m$ observes its neighborhood $N(m)$. 
For each cell $c \in N(m)$, it determines whether $c$ is empty or occupied, and if occupied, it observes the state of the module in $c$.  
We assume that all modules have the same sense of chirality.
In the \emph{compute} phase, $m$ may transition to a new state.
It may also change the states of its neighbors, based solely on the information gathered during the look phase.
In the \emph{move} phase, the module performs an action associated with its current state, which is either a slide, a convex transition, or staying in place.

\begin{figure}[tbp]
    \centering
    \begin{subfigure}[b]{0.333\linewidth}
        \centering%
        \includegraphics[width=\linewidth,page=1]{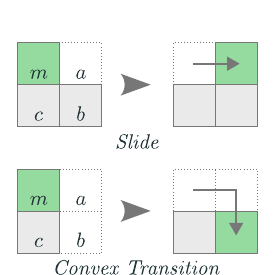}%
        \subcaption{}
        \label{subfig:def-moves}
    \end{subfigure}%
    \begin{subfigure}[b]{0.333\linewidth}
        \centering%
        \includegraphics[width=\linewidth,page=2]{definitions}%
        \subcaption{}
        \label{subfig:def-layers}
    \end{subfigure}%
    \begin{subfigure}[b]{0.333\linewidth}
        \centering%
        \includegraphics[width=\linewidth,page=3]{definitions}%
        \subcaption{}
        \label{subfig:def-rhombus}
    \end{subfigure}%
    \caption{(a) Modules move only via slides and convex transitions. For a slide to be viable, cell $a$ must be empty; for a convex transition to be viable, both $a$ and $b$ must be empty. (b) Example configuration with the leader module in the center. Modules of equal Manhattan distance from the leader are in the same layer $L_i$, depicted by having the same shade of grey. (c) The rhombus $\rhombus = (v_0, \ldots, v_{n-1})$ where $n = 20$. The connector cell of each layer is marked with a dot.}
\end{figure}

Conflicts may arise during the compute phase, when multiple modules attempt to update the state of the same neighbor, or during the move phase, if movement trajectories intersect.
We resolve both types of conflicts arbitrarily: a module may take on any of the proposed new states, and in the case of movement conflicts, an arbitrary module proceeds while the others remain in place.
While such conflicts could, in principle, lead to a livelock (e.g., if the conflicts are cyclic), our algorithm is designed to ensure that progress is always made.

We assume that the initial configuration $(M_0, s_0, l)$ satisfies the following:
The initial set of all modules $M_0$ is side-connected.
Exactly one module occupies the leader cell $l \in M_0$.
All other modules share the same initial state, i.e., $s_0(m) = s_0(m')$ for all $m, m' \in M_0 \setminus \{l\}$.
The leader's state is distinct from all other modules, i.e., $s_0(l) \neq s_0(m)$ for all $m \in M_0$.

By executing a distributed algorithm, the system goes through a sequence of configurations $(M_0, s_0, l), (M_1, s_1, l)\ldots, (M_T, s_T, l)$ where each pair of consecutive configurations corresponds to the execution of one synchronous round.
We require that in each round, the subset of modules that do not move remains side-connected.
The goal is to reach a final configuration $(M_T, s_T, l)$ where each module is in a terminated state, and the sum $\sum_{m \in M_T} d(m,l)$ is minimized, where $d(m,l)$ denotes the Manhattan ($L_1$) distance between $m$ and $l$.
Minimizing this sum in the square grid naturally results in modules being arranged in the shape of a rhombus, centered at $l$, whose outermost layer may be partially filled.

While modules are technically finite automata, we describe our distributed algorithm from a higher-level of abstraction, assuming each module stores only a constant number of variables, each of constant size domain.

\paragraph*{Additional Terminology.}

Denote by $d(c_1, c_2)$ the Manhattan distance between two cells $c_1, c_2 \in C$.
A \emph{layer} $L_i \subset C$ consists of all cells at Manhattan distance $i$ from the leader cell $l$, i.e., $L_i = \{c \in C \mid d(l,c) = i\}$ (see \cref{subfig:def-layers}).
For each layer, we define a unique \emph{connector} cell.
The connector of $L_0 $ is $l$, and for any $i > 0$, the connector of $L_i$ is the cell north of the westernmost cell in~$L_{i-1}$.

Our goal is to form a rhombus shape consisting of consecutive layers arranged around the leader module.
Formally, we define a \emph{rhombus} $\rhombus{}=(v_0,...,v_{n-1})$ of size $n$, which is constructed by appending the cells of each layer clockwise, starting from their respective connector cells (see \cref{subfig:def-rhombus}).
Our algorithm aims to fill each cell in $\rhombus{}$ with a module in this sequential order.

A \emph{path} (or \emph{corner-path}) of length $k$ from cell $c_0$ to cell $c_k$ is a tuple $(c_0, c_1, ..., c_k)$ such that $c_i$ is side-adjacent (or adjacent) to $c_{i+1}$ for all $0 \leq i < k$.
A set of cells is \emph{side-connected} (or \emph{corner-connected}) if there is a path (or corner-path) between any two cells from that set, and that path is fully contained within the set.
We call a module $m \in M$ \emph{removable} if $M \setminus \{m\}$ is side-connected.
For simplicity, we slightly abuse notation and treat tuples as sets when referring to their contents, e.g., we write $m \in P$ if module $m$ appears in path $P$.

A \emph{hole} $\hole{}$ is a maximal corner-connected subset of empty cells, and a \emph{pseudo-hole} $\phole{}$ is a maximal side-connected subset of empty cells.
Since the set of modules is finite, there is exactly one infinite hole and one infinite pseudo-hole.
For example, the configuration in \cref{subfig:def-layers} has one finite hole of size three, and two finite pseudo-holes of sizes two and three.

Any hole can be partitioned into one or more pseudo-holes.
Two distinct pseudo-holes $\phole{}$ and $\phole{}'$ are \emph{adjacent} if there exist adjacent cells $e \in \phole{}, e' \in \phole{}'$.
In that case, their common neighborhood $N(e) \cap N(e')$ contains precisely two occupied cells $c, c'$.
We refer to $(c,c')$ as a \emph{critical pair}, to cell $e$ as the \emph{bridge} from $\phole{}$ to $\phole{}'$, and to cell $e'$ as the \emph{bridge} from $\phole{}'$ to $\phole{}$.
Notably, for any two pseudoholes $\phole{}, \phole{}'$ there can be at most one critical pair $(c,c')$, and consequently, at most one bridge in each direction.
Otherwise, $\phole{} \cup \phole{}'$ would contain a corner-path cycle fully enclosing either $c$ or $c'$, contradicting that $M$ is side-connected. 

Finally, the \emph{boundary} of a set of cells $V$ is defined as $B(V) = (\bigcup_{v \in V} N(v)) \setminus V$. 
Note that by definition, the boundary of a hole can only contain occupied cells, while the boundary of a pseudo-hole may also contain empty cells that belong to adjacent pseudo-holes.
\section{A Distributed Rhombus Formation Algorithm}
\label{sec:algorithm}

\newcommand{\pActive}{\textup{\textsc{Active}}}
\newcommand{\pPassive}{\textup{\textsc{Passive}}}
\newcommand{\pHead}{\textup{\textsc{Head}}}
\newcommand{\pTail}{\textup{\textsc{Tail}}}

In this section, we describe our distributed rhombus formation algorithm.
\iftoggle{fullversion}{
    Some technical details are deferred to the end of this section (\cref{subsec:technical-details}).
}{
    Some technical details, including how the next rhombus cell is determined from local observations, message handling, and the handling of one corner case, are deferred to Appendix~\ref{sec:appendix-technical-details}.
}
\cref{fig:example} illustrates the execution of the algorithm on an example configuration.

\begin{figure}[htb]
    \centering
    \begin{subfigure}[b]{0.5\linewidth}
        \centering%
        \includegraphics[width=\linewidth,page=1]{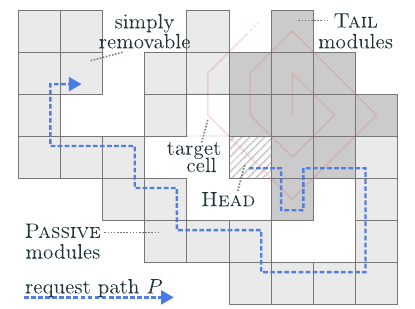}%
        \subcaption{}
        \label{subfig:example-1}
    \end{subfigure}%
    \begin{subfigure}[b]{0.5\linewidth}
        \centering%
        \includegraphics[width=\linewidth,page=2]{algo-example}%
        \subcaption{}
        \label{subfig:example-2}
    \end{subfigure}\\
    \begin{subfigure}[b]{0.5\linewidth}
        \centering%
        \includegraphics[width=\linewidth,page=3]{algo-example}%
        \subcaption{}
        \label{subfig:example-3}
    \end{subfigure}%
    \begin{subfigure}[b]{0.5\linewidth}
        \centering%
        \includegraphics[width=\linewidth,page=4]{algo-example}%
        \subcaption{}
        \label{subfig:example-4}
    \end{subfigure}%
    \caption{
    (a) The \pHead{} initiates a request, which is forwarded anti-clockwise along the target boundary to a simply removable module.
    (b) The module becomes \pActive{}, moves along the request path $P$ in reverse until blocked by a critical pair, then becomes \pPassive{} as the blocking module is not within $P$.
    (c) The request continues until the next forwarding direction is not a \pPassive{} module.
    Going backwards in the request path, the first conditionally removable module becomes \pActive{}.
    (d) The \pActive{} module is blocked again, this time the blocking module is within $P$, and its successor equals its predecessor.
    The \pActive{} module sends a cleanup message to erase the requests stored at modules that it cannot reach.
    Afterwards, it moves directly to the target cell (not visualized).
    }
    \label{fig:example}
\end{figure}

In our algorithm, modules transition through four distinct phases: \pHead{}, \pTail{}, \pActive{}, and \pPassive{}.
From a high level, the \pHead{} module identifies the next cell to fill in the rhombus, the \pActive{} module moves to occupy this cell, \pTail{} modules have already permanently joined the rhombus, and \pPassive{} modules include all modules that have yet to join.
Initially, the leader module is the \pHead{} while all other modules are \pPassive{}.
In any given round, there is precisely one \pHead{} and at most one \pActive{} module.
Any module that is not \pActive{} remains stationary and solely participates in communication.

The \pHead{} always identifies the next cell to occupy, referred to as the \emph{target cell}, and stores the direction to this cell in its constant memory.
In most cases, the target cell is the immediate successor of the current \pHead{} within the rhombus $\rhombus{}$ (except for one edge case discussed in \iftoggle{fullversion}{\cref{subsec:technical-details}}{Appendix~\ref{sec:appendix-technical-details}}).
If the target cell is already occupied by some module, then that module becomes the new \pHead{} while the previous \pHead{} becomes a \pTail{}.
If it is empty, the \pHead{} initiates a request message to find a \pPassive{} module that can become \pActive{} and move towards the target cell.
Messages are only exchanged between side-adjacent modules, and they travel along the boundary of the hole that contains the target cell, which we refer to as \emph{target boundary} and \emph{target hole}, respectively.
Moreover, messages carry no information besides their type.
A request message always travels anti-clockwise along the target boundary.
If a request fully traverses the target boundary without seeing \pPassive{} modules, the \pHead{} terminates.
All other modules then terminate in subsequent rounds upon observing a terminated neighbor.
Otherwise, if a request reaches a \pPassive{} module $m$, that module locally checks whether moving may break connectivity, which we formalize as follows:

\begin{definition}
    \label{def:simply-removable}
    A module $m$ is \emph{simply removable} if (1) it is side-adjacent to the target hole, and (2) for every pair $m_1, m_2$ of modules side-adjacent to $m$, there exists a path from $m_1$ to $m_2$ within $N(m) \cap M$. \emph{(Note that $m \notin N(m)$.)}
\end{definition}

If $m$ is simply removable, it becomes \pActive{}; otherwise, it forwards the request anti-clockwise along the target boundary.
Note that simple removability is a sufficient but not necessary condition for a module to be removable.
It may occur that no \pPassive{} module that receives a request is simply removable, in which case local information alone does not suffice to identify a removable module.
Modules maintain an ordered list of the requests they receive, together with the direction from which each request was received.
A module can receive at most four requests for the same target cell, one for each side-adjacent module.
Additionally, we ensure that all memory is reset before a new request is initiated.
Hence, constant memory suffices to store all requests.

\begin{definition}
    \label{def:request-path}
    A \emph{request path} $P = (m_1, m_2, ..., m_k)$ is a maximal directed path along the target boundary such that for each $1 < i \leq k$, module $m_i$ stores a request (or receives a request at the start of the current round) from module $m_{i-1}$, and each directed edge $(m_{i-1}, m_i)$ appears at most once in $P$.
\end{definition}

\begin{definition}
    \label{def:passive-segment}
    Given a request path $P = (m_1, m_2, ..., m_k)$, a \emph{passive segment} is a maximal contiguous sub-path $S = (m_i, ..., m_j) \subseteq P$, where $1 < i \leq j \leq k$, consisting entirely of \pPassive{} modules.
\end{definition}

Note that the request path is not necessarily simple in terms of nodes, but it contains any directed cycle at most once.
In contrast, in the analysis we show that following our algorithm, the request path $P$ contains at most one passive segment $S$, that this segment is simple (in terms of nodes), and that all \pPassive{} modules storing a request belong to $S$.
As a result, any \pPassive{} module $m$ that receives a request in some round must be the last module on both $P$ and $S$ in that round, i.e., $m = m_k = m_j$.
Since $m$ can observe the state of its neighbors, it knows which of them belong to the request path.
We use this additional information to identify simple cycles within the set of all modules $M$.
This gives us an additional criterion to check for removability: 

\begin{definition}
    \label{def:conditionally-removable}
    Let $Z = (m_0, m_1, ..., m_{k-1})$ be a simple cycle within the set of modules $M$. A module $m_i$ with $0 \leq i < k$ is \emph{conditionally removable}, if (1) $m_i$ is side-adjacent to the target hole, and (2) for every module $m$ side-adjacent to $m_i$, there exists a path from $m$ to $m_{i-1}$ or a path from $m$ to $m_{i+1}$ within $N(m) \cap M$ (where indices are taken modulo $k$).
\end{definition}

We now describe the procedure by which a conditionally removable module is selected to become \pActive{}.
The procedure is initiated whenever the last module $m_k$ on the request path $P$ is \pPassive{} and the module in the next forwarding direction is already contained in $P$ or is not a \pPassive{} module.
There are two cases to consider: 

\emph{Case 1: Request cycles back to itself.} 
Suppose that the module in the next forwarding direction from $m_k$ is a \pPassive{} neighbor $m_l\in P$, i.e., a module that already belongs to the request path $P$ with $l < k$.
Since the passive segment $S$ is unique and simple, $(m_l, m_{l+1}, ..., m_k)$ is a simple cycle.
In this case, the algorithm proceeds as follows:
Starting with module $m_k$, traverse the request path in reverse order.
At each step $i = k, k-1, \ldots, l$, the current module $m_i$ checks whether it is conditionally removable.
If so, it becomes \pActive{} and the process stops.
Otherwise, it deletes the request it has received from $m_{i-1}$ from its memory and sends an activation message to its predecessor $m_{i-1}$, which then continues the process.
In the analysis, we show that this process always finds a conditionally removable module before the activation message reaches $m_l$.

\emph{Case 2: Request reaches a non-\pPassive{} module.}
Now suppose that the module in the next forwarding direction from $m_k$ is not \pPassive{}, i.e., a \pHead{} or \pTail{} module.
Since $m_k$ is \pPassive{}, the passive segment $S = (m_i, ..., m_k)$ is non-empty.
Moreover, by maximality of $S$, the module $m_{i-1}$ must also be non-\pPassive{}.
Let $P'$ be a shortest path from $m$ to $m_{i-1}$ within the set of all non-\pPassive{} modules.
We show in \iftoggle{fullversion}{\cref{lem:extended-rhombus-weakly-convex}}{\cref{lem-appendix:extended-rhombus-weakly-convex}} that such a path exists. 
Since all shortest paths are simple and $S$ contains only \pPassive{} modules, the paths $S$ and $P'$ are disjoint.
Therefore, their concatenation forms a simple cycle.
In this case, the algorithm proceeds identical to Case~1.
Starting at $m_k$, each module along $S$ is checked in reverse order for conditional removability while deleting its stored request.
The first such module becomes \pActive{} and stops the process.
We show in the analysis that this process always finds a conditionally removable module before the activation message reaches $m_{i-1}$ unless the target hole has size one, a special case addressed in \iftoggle{fullversion}{\cref{par:special-case}}{Appendix~\ref{sec:appendix-technical-details}}.

Eventually, some \pPassive{} module becomes \pActive{} and deletes its stored request. 
The \pActive{} module moves via slides and convex transitions through the target hole, following the request path $P$ in reverse direction. 
During this process, it deletes requests stored at adjacent modules such that it always remains side-adjacent to the last module of $P$.

For simplicity, we treat the cell occupied by the \pActive{} module as empty w.r.t. all definitions and terminology. 
For example, if after moving it is side-adjacent to cells $c_1$ and $c_2$ that formed a critical pair in the previous round, we still consider $(c_1, c_2)$ a critical pair, even though they are now technically connected through the \pActive{} module.

\subsection{Moving past Critical Pairs}
\label{par:critical-pair-cases}
It may happen that the movement of the \pActive{} module is blocked by a critical pair $(m_k, c)$, where $m_k$ is the last module on the request path $P$.  
At this point, the current position $m$ of the \pActive{} module is one of the bridge cells of the critical pair.  
Let $b \neq m$ be the other bridge cell.  
By placing a module at $b$, we obtain a hole $\hole_m$ containing $m$, and by placing a module at $m$, we obtain a hole $\hole_b$ containing $b$ (see \cref{fig:active-critical-pair-algo}).  
Since any two adjacent pseudo-holes are separated by a unique critical pair, the holes $\hole_m$ and $\hole_b$ are disjoint.  
As a result, exactly one of them contains the target cell.

    \begin{figure}[tbp]
        \centering
        \begin{subfigure}[b]{0.25\linewidth}
            \centering%
            \includegraphics[width=\linewidth,page=2]{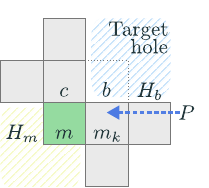}%
            \subcaption{}
            \label{subfig:active-critical-pair-algo-1}
        \end{subfigure}%
        \begin{subfigure}[b]{0.25\linewidth}
            \centering%
            \includegraphics[width=\linewidth,page=4]{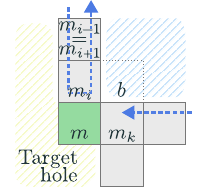}%
            \subcaption{}
            \label{subfig:active-critical-pair-algo-2}
        \end{subfigure}%
        \begin{subfigure}[b]{0.25\linewidth}
            \centering%
            \includegraphics[width=\linewidth,page=5]{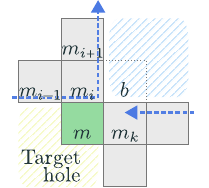}%
            \subcaption{}
            \label{subfig:active-critical-pair-algo-3}
        \end{subfigure}%
        \begin{subfigure}[b]{0.25\linewidth}
            \centering%
            \includegraphics[width=\linewidth,page=6]{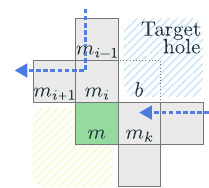}%
            \subcaption{}
            \label{subfig:active-critical-pair-algo-4}
        \end{subfigure}%
        \caption{Case distinction for determining which pseudo-hole contains the target cell when the \pActive{} module $m$ is blocked by a critical pair $(m_k, c)$. The decision is based on whether $c$ occurs on the request path and, if so, on the relative position of its successor and predecessor. A variant of (b) where $m_{i+1} = m_{i-1}$ is corner-adjacent to $m$, as well as the case where $m_i = m_1$, are omitted.}
        \label{fig:active-critical-pair-algo}
    \end{figure}

To proceed, the \pActive{} module determines whether $\hole_m$ or $\hole_b$ contains the target cell, a decision that can be made locally based on the following case distinction:  
If $c \notin P$ (see \cref{subfig:active-critical-pair-algo-1}), then the target cell lies in $\hole_b$.
Otherwise, let $m_i \in P$ be the last occurrence of $c$ on the request path, i.e., the module $m_i = c$ with maximum index $i$.
Since requests are stored in ordered lists, this information can be read from the module's memory. 

If $m_i = m_1$, i.e., it has no predecessor and thus initiated the request, then $m$ can read the position of the target cell from the memory of $m_i$ and directly determine whether it lies in $\hole_m$ or $\hole_b$.
Since modules share a common chirality, the request can carry the \pHead{}'s orientation, so $m$ and $m_i$ agree on the direction of north. 

Otherwise, $m_i$ has a predecessor $m_{i-1} \in P$ and a successor $m_{i+1} \in P$.
Then the target cell lies in $\hole_m$ if $m_{i+1} = m_{i-1}$ or $m_{i+1} \notin N(m)$ (see \cref{subfig:active-critical-pair-algo-2} and \cref{subfig:active-critical-pair-algo-3}, respectively); otherwise it lies in $\hole_b$ (see \cref{subfig:active-critical-pair-algo-4}).
We prove claims in the proof of \cref{lem:invariants}.

In our algorithm, the \pActive{} module $m$ proceeds as follows:
It first determines, using the above case distinction, whether $\hole_m$ or $\hole_b$ contains the target cell.
If $\hole_m$ contains the target cell, then $m$ sends a cleanup message to $m_k$ (e.g., see \cref{subfig:example-4}).
Each module receiving this message forwards it to its predecessor on the request path and then deletes its last stored request.
Eventually, the message reaches $c$, at which point it is dropped and $m$ can resume its movement along the boundary of $\hole_m$.
Otherwise, if $\hole_b$ contains the target cell, then $m$ becomes \pPassive{} and acts as if it had received a request from $m_k$ (e.g., see \cref{subfig:example-2,subfig:example-3}).
Note that by becoming \pPassive{} and continuing the request from $m_k$, all pseudo-holes contained in $\hole_m$ are pruned from the target hole.
This behavior is essential for showing that all requests terminate in linear time.

\iftoggle{fullversion}{
    \subsection{Technical Details}
    \label{subsec:technical-details}

    In the following, we describe the technical details omitted from the previous section to complete the description of our rhombus formation algorithm.

    We first describe how a \pHead{} module determines the next target cell.
    Let $\rhombus{} = (v_0,...,v_{n-1})$ be a rhombus of size $n$, and let $v_i$ be the cell in $\rhombus{}$ occupied by the \pHead{} module.
    In our algorithm, we maintain as an invariant that $\{v_0, ..., v_{i-1}\}$ is precisely the set of all cells occupied by \pTail{} modules.
    This ensures that the \pHead{} module can determine its position within the rhombus by observing the states of its neighbors: whether it is at the first cell ($v_i = v_0$), at an extremal cell (e.g., northernmost), on one of the layer's straight segments, or at its connector.
    For example, if there are \pTail{} neighbors in direction $\N$ and $\E$ (e.g., \cref{subfig:example-1}), then $v_i$ lies on a straight segment in the rhombus' southwest quadrant, and the target cell is in the $\NW$ direction.

    Next, we clarify how modules determine the direction in which messages are forwarded and whether they are side-adjacent to the target hole.
    We assume that whenever a module $m$ receives a message, it knows which of its neighbors $src$ has sent the message.
    This allows $m$ to compute both the direction of the target hole and the next forwarding direction along the target boundary.
    Module $m$ scans its neighbors in anti-clockwise order, starting with module $src$.
    The first empty cell following $src$ in this scan belongs to the target hole, and the first side-adjacent module following $src$ is the next module along the target boundary in anti-clockwise direction.
    Similarly, when initiating a request, module $m$ performs a scan starting with the target cell instead of $src$.
    When forwarding activation and cleanup messages, a module can directly read the forwarding direction from its stored requests.


    \paragraph*{Holes of Size One}
    \phantomsection
    \label{par:special-case}
    In this paragraph, we discuss the special case where an activation message does not reach a conditionally removable module.
    Let $P = (m_1, ..., m_k)$ be the request path, $S = (m_i, ..., m_k)$ with $i > 1$ its passive segment, and let $m$ be a non-\pPassive{} module side-adjacent to $m_k$.
    Recall that we obtain a simple cycle by concatenating $S$ with a shortest path from $m$ to $m_{i-1}$ within the set of all non-\pPassive{} modules.
    In general, any activation message traversing $P$ in reverse reaches a conditionally removable module before it reaches $m_{i-1}$.
    However, this may not be the case if the target hole has size one, i.e., it consists of only the target cell.
    There are configurations in which no \pPassive{} module on the target boundary is conditionally removable (e.g., see \cref{subfig:special-case1}).
    In this case, the activation message fully traverses the request path in reverse until it reaches the \pHead{} module.

    Although no \pPassive{} module is initially conditionally removable, we can still create a conditionally removable module using a second request.
    Let $e$ be the target cell, and assume that the target hole has size one and that the target boundary contains no conditionally removable \pPassive{} module.
    Let $Z$ be the simple cycle consisting of the eight modules in the target boundary in an anti-clockwise order starting at the \pHead{} module.
    Let $m_1, m_2$ be the first two \pPassive{} modules in $Z$ that are side-adjacent to the target hole.
    Since neither $m_1$ nor $m_2$ is conditionally removable, there exists two empty cells $e_1$ and $e_2$ corner-adjacent to $m_1$ and $m_2$, respectively, such that $e_1$ and $e_2$ are corner-adjacent to each other and thus lie in the same hole (see \cref{subfig:special-case1}).

    Once the \pHead{} module receives an activation message, it instructs $m_1$ to initiate a new request for cell $e_1$.
    At this point, $m_1$ temporarily acts as the \pHead{}, and all other modules within $Z$ temporarily act as \pTail{} modules.
    In the analysis, we show that any request results in a series of module movements that either fills the target cell or reduces the target hole to size one, without changing the position of any non-\pPassive{} module.
    This process leads to one of two outcomes (e.g., see \cref{subfig:special-case2,subfig:special-case3}):
    (1) Cell $e_1$ is filled by some module.
    Then $m_1$ becomes conditionally removable w.r.t. $Z$ and can slide into cell $e$.
    (2) All cells in the neighborhood of $e_1$ are filled by modules while $e_1$ remains empty.
    As a result, cell $e_2$ must be occupied by a module, which implies that $m_2$ is now conditionally removable w.r.t. $Z$ and can slide into cell $e$.
    In either case, a module moves into cell $e$, after which all modules adjacent to $e$ revert to their original roles and the algorithm proceeds as usual.

    In conclusion, for the purpose of the analysis, it suffices to focus on holes of size at least two:
    Once the target hole $e$ is reduced to size one, at most one additional request for some intermediate cell $e_1 \neq e$ is initiated.
    Whether this second request successfully fills the cell $e_1$ is irrelevant; in both cases, there is a conditionally removable module that can slide into $e$.

    \begin{figure}[t]
        \centering
        \begin{subfigure}[b]{0.33\linewidth}
            \centering%
            \includegraphics[width=\linewidth,page=1]{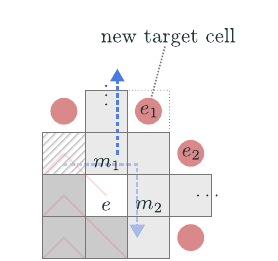}%
            \subcaption{}
            \label{subfig:special-case1}
        \end{subfigure}%
        \begin{subfigure}[b]{0.33\linewidth}
            \centering%
            \includegraphics[width=\linewidth,page=2]{special-case}%
            \subcaption{}
            \label{subfig:special-case2}
        \end{subfigure}%
        \begin{subfigure}[b]{0.33\linewidth}
            \centering%
            \includegraphics[width=\linewidth,page=3]{special-case}%
            \subcaption{}
            \label{subfig:special-case3}
        \end{subfigure}%
        \caption{Example configuration with a target hole of size one whose boundary has no conditionally removable module. Empty cells are marked by red dots; other cells may be empty or occupied. (a) Module $m_1$ initiates a request for $e_1$. (b) If $e_1$ is filled, $m_1$ becomes conditionally removable and slides into $e$. (c) If not, and the case repeats, $m_2$ becomes conditionally removable and slides into $e$.}
        \label{fig:special-case}
    \end{figure}
}{}
\section{Analysis}
\label{sec:analysis}

In the following, we prove that our distributed rhombus formation algorithm reconfigures any initially side-connected configuration $(M_0, s_0, l)$ with a leader in cell $l \in M_0$ into a rhombus centered at $l$, with its outermost layer possibly partially filled.
\iftoggle{fullversion}{}{
For clarity, we sketch the key proofs in this section and defer technical lemmas and full proofs to Appendix~\ref{sec:appendix-proofs}.
}

\begin{lemma}
    \label{lem:retired-is-rhombus}
    Let $\rhombus{} = (v_0, v_1,\ldots, v_{n-1})$ be the rhombus centered at $v_0 = l$.
    If the \pHead{} occupies cell $v_i$ where $0 \leq i < n$, then the set of all \pTail{} modules is exactly $\{v_0, \ldots, v_{i-1}\}$.
\end{lemma}

\begin{proof}
    Since the initial position of the \pHead{} module is $v_0$ and there are no \pTail{} modules, the lemma holds initially.
    Now, assume the lemma holds at the start of an arbitrary round.
    There is exactly one \pHead{} module at all times, and all \pTail{} modules are contained within $\{v_0, \ldots, v_{i-1}\}$.
    This implies that cell $v_{i}$ must be empty or occupied by a \pPassive{} or \pActive{} module.
    As the \pHead{} module does not move, and \pTail{} modules neither move nor change phase, the lemma holds until $v_i$ is occupied.
    Once $v_{i}$ is occupied by a \pPassive{} or \pActive{} module, that module becomes the new \pHead{} while the previous \pHead{} transitions to a \pTail{}.
    Thus, in the next round, the \pHead{} module occupies cell $v_{i}$ and the set of all cells occupied by a \pTail{} module expands to $\{v_0, \ldots, v_{i}\}$.
    The lemma follows inductively.
\end{proof}

The following lemma was proven by Dumitrescu and Pach in~\cite{Dumitrescu2004-PushingSquaresAround}.
\iftoggle{fullversion}{
    For ease of reference and to keep the presentation self-contained, we state and reprove it here using our terminology.
}{
    For ease of reference, we state it here using our terminology.
}

\begin{lemma}
    \label{lem:move-in-phole}
    Let $\phole{}$ be an arbitrary pseudo-hole, let $a, b \in \phole{}$ be any distinct cells side-adjacent to its boundary $B(\phole{})$.
    A module placed at $a$ can move into cell $b$ while remaining side-adjacent to cells in $B(\phole{})$ throughout its movement.
\end{lemma}

\iftoggle{fullversion}{
    \begin{proof}
        Let $O$ be the point set obtained by intersecting the union of all cells in $\phole{}$ with the union of all cells in $B(\phole{})$.
        We refer to $O$ as the \emph{outline} of the pseudo-hole $\phole{}$.
        Since pseudo-holes are maximal and side-connected, $O$ is a simple polygon.
        
        Let $c, d \in B(\phole{})$ be any two occupied cells side-adjacent to $a$ and $b$, respectively.
        Again, by maximality of $\phole{}$, $c$ and $d$ must exist.

        Since $O$ is a simple polygon, we can define a closed simple curve $\mathcal{C}$ on points within $O$ that starts and ends at the midpoint of the common side of cells $a$ and $c$.
        Let $S = (s_0, s_1, \ldots, s_m)$ be the sequence of sides ordered according to their appearance in $\mathcal{C}$, and let $W = (w_0, w_1,\ldots, w_m)$ be the sequence of empty cells such that cell $w_i$ has side $s_i$ for all $0 \leq i \leq m$. 
        Note that cell $w_i$ is uniquely defined for all $i$ with $0\leq i \leq m$ since each side $s_i$ belongs to precisely two cells, one of which is occupied and lies in $B(\phole{})$, the other is the empty cell $w_i \in \phole{}$.
        Since $b \in \phole{}$ and $d \in B(\phole{})$, the common side of $b$ and $d$ must be fully contained in $\mathcal{C}$.
        Consequently, there is a $k$ with $0 \leq k \leq m$ such that $w_k = b$.

        A module placed at $a$ moves from cell $a = w_0$ to cell $b = w_k$ as follows:
        For each $i$ with $0 \leq i < k$: 
        If $w_{i+1} = w_{i}$, then the module stays.
        Otherwise, if $s_{i+1}$ is collinear to $s_i$, then the module slides to $w_{i+1}$.
        Otherwise, it performs a convex transition into cell $w_{i+1}$.
        After staying or moving $k$ times, the module is side-adjacent to $w_k = b$.
    \end{proof}
}{}

We now establish essential invariants that hold throughout the execution of the algorithm.
These invariants capture the structure of the request path (see \cref{def:request-path}) and its passive segment (see \cref{def:passive-segment}), and are essential for the algorithm's correctness.

\begin{lemma}
    \label{lem:invariants}
    Let $P = (m_1, m_2, \ldots, m_k)$ be the request path.
    In every round during the execution of the rhombus formation algorithm, the following properties hold as invariants:
    \begin{enumerate}
        \item For every module $m$ that stores a request from some module $m'$, the directed edge $(m', m)$ is contained in $P$.
        \item The request path $P$ contains at most one passive segment $S = (m_i, ..., m_j)$, and if $S$ is non-empty, then $j = k$, i.e., the passive segment ends at the last module of $P$.
        \item The passive segment $S$ is simple, i.e., no module appears more than once in $S$.
    \end{enumerate}
\end{lemma}

\iftoggle{fullversion}{
    \begin{proof}
        The invariants are shown by induction over the rounds of the algorithm.
        
        Initially, no request has been initiated, and no module stores any request.
        Hence, all properties hold trivially.
        Now suppose that module $m_1$ initiates a request, i.e., $P = (m_1)$ and the passive segment $S$ is empty.  
        Since no module has yet stored any request, Property~(1) holds.
        Since $S$ is empty, Properties~(2) and (3) hold trivially as well.

        Now consider an arbitrary round in which $P$ is not empty and assume that the invariants hold.
        We distinguish the possible actions that may occur in the current round:

        First, assume module $m_k$ forwards a request to some module $m_{k+1}$.
        A \pPassive{} module only forwards a request if $m_{k+1}$ is a \pPassive{} module and not contained in $P$.
        Hence, $P$ and $S$ extend by a \pPassive{} module $m_{k+1}$ previously not contained in $P$ and $S$, and all properties hold.
        Whenever a non-\pPassive{} module forwards a request, then by (2), the passive segment is empty in that round.
        Hence, afterwards, the passive segment remains empty or consists only of $m_{k+1}$ and the properties hold.

        Second, a \pPassive{} module $m$ becomes \pActive{} upon receiving a request or activation message.
        In the first case, it must be the last module in $P$ by invariant (2).
        In the second case, the activation message was received from its successor in $P$ w.r.t. the previous round.
        Only \pPassive{} modules send activation messages, and after doing so, they delete their last stored request.
        By invariant (3), the passive segment $S$ is simple.
        Therefore, the successor of $m$ w.r.t. the previous round has deleted its only stored request and is no longer in $P$.
        Hence, in both cases $m$ is the last module in $P$, i.e., $m = m_k$, and it deletes its request upon becoming \pActive{}.
        It follows that the invariants carry over from the previous round, since only the last edge in $P$ is removed.

        \begin{figure}[t]
            \centering
            \includegraphics[width=0.5\linewidth,page=1]{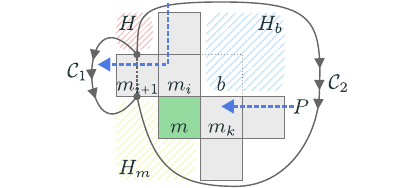}%
            \caption{An \pActive{} module $m$ side-adjacent to a critical pair $(m_k, m_i)$ with bridge cells $m$ and $b$. Module $m_i$ is a predecessor of $m_k$ on the request path, and its direct successor $m_{i+1}$ is contained in $N(m)$. In this case $\hole_b = \hole$ contains the target cell.}
            \label{fig:active-critical-pair}
        \end{figure}

        Third, consider the deletion of requests by an \pActive{} module, either directly or by its cleanup messages.
        As shown above, only the last module in $P$ can become \pActive{} and afterwards deletes its stored request, which corresponds to the last edge in $P$.
        It follows that an \pActive{} module $m$ always begins its movement side-adjacent to the last module in $P$.  
        It then moves by following $P$ in reverse.
        As long as it stays within the same pseudo-hole, by \cref{lem:move-in-phole} it can remain side-adjacent to the last module in $P$ while iteratively deleting the request corresponding to the last edge in $P$. 
        Whenever the last edge of $P$ is deleted, the invariants directly carry over from the previous round. 

        If the movement of $m$ is blocked by a critical pair, its behavior changes: the module either becomes \pPassive{} or initiates a cleanup message.  
        To prove that the invariants are still maintained in this case, we must show that the case distinction in \cref{par:critical-pair-cases} correctly determines which of the pseudo-holes adjacent to $m$ is part of the target hole.
        
        Let $(m_k, c)$ be the critical pair where $m_k$ is the last module on $P$.
        Let $b$ be the bridge cell adjacent to that pair, where $b \neq m$.
        Recall that we treat the \pActive{} module as an empty cell w.r.t. our terminology.
        Hence, $m$ is the other bridge cell adjacent to the critical pair.
        Let $\hole{}_m$ be the hole that contains cell $m$, if $b$ would be occupied, and let $\hole_b$ be the hole that contains cell $b$, if $m$ would be occupied (see \cref{fig:active-critical-pair}).
        W.l.o.g., we assume that $c$ is in direction $\N$ of $m$.
        The other four cases are analogous up to rotation.

        First, consider the case where $c$ is contained in $P$.
        Consider the last appearance of $c$ in $P$, i.e., maximize the index $i$ such that $m_i = c$.
        Consider an anti-clockwise scan around $m_1$ starting from the target cell $e$ and note that $m_2$ is the first side-adjacent module in that scan after the target cell.
        It follows that in the anti-clockwise scan around $m_1$, any empty cell between $e$ and $m_2$ belongs to the target hole, and the last of those empty cells is adjacent to $m_2$.
        Continue this process with $m_2, m_3$ and so forth until module $m_i$.
        Note that its successor $m_{i+1}$ must be distinct from $m_k$, since $m_k$ and $m_i$ are not side-adjacent, i.e., $k > i +1$.
        If $m_{i+1} \notin N(m)$ or $m_{i+1} = m_{i-1}$, then the scan around $m_i$ from $m_{i-1}$ to $m_{i+1}$ contains an empty cell side-adjacent to $m$.
        It follows that $\hole_m$ contains the target cell.
        Consider the case where $m_{i+1} \in N(m)$ and $m_{i+1} \neq m_{i-1}$.
        Since $i > 1$ by assumption and $b$ is empty, $m_{i-1}$ must lie in direction $\N$ of $m_i$.
        In this case, the only empty cell in the scan around $m_i$ from $m_{i-1}$ to $m_{i+1}$ is the cell in direction $\NW$ of $m_i$.
        Denote $\hole{}$ the hole containing that cell, and assume by contradiction that $\hole = \hole_m$.
        Consider a simple directed curve that starts and ends at the midpoints of the northern and southern side of $m_{i+1}$, respectively, and that is entirely contained within $\hole$.
        Since holes are defined w.r.t. corner-adjacency, such a curve must exist if $\hole = \hole_m$.
        Now connect the endpoints of that curve to a simple directed loop.
        Depending on the positioning of modules, that loop is either anti-clockwise or clockwise (see $\mathcal{C}_1$ and $\mathcal{C}_2$ in \cref{fig:active-critical-pair}).
        If the loop is anti-clockwise, then it encloses all successors of $m_{i+1}$, but not $m_k$.
        Otherwise, if the loop is clockwise, then it encloses no successor of $m_{i+1}$, but it encloses $m_k$.
        Both contradict that $m_k$ is a successor of $m_{i+1}$.
        Hence, it holds that $\hole{} \neq \hole_m$.
        Note that by placing a module at $m$, the target hole disconnects into precisely two distinct holes.
        Otherwise, $m_k$ would have only one side-adjacent neighbor, which implies that it is simply removable and contradicts that $m$ became \pActive{} before $m_k$.
        It follows that $\hole = \hole_b$.

        Second, consider the case where $c$ is not contained in $P$.
        Then $m_k$ is the first module in $P$ that is side-adjacent to $m$.
        This again follows since $m_k$ must have at least two neighbors in the round in which $m$ became \pActive{}, otherwise contradicting that $m$ becomes \pActive{} before $m_k$.
        Hence, $\hole_b$ must contain the target cell.
        
        We have shown that the \pActive{} module can determine whether $\hole_m$ or $\hole_b$ contains the target cell.
        If $\hole_m$ contains the target cell, then it sends a cleanup message to $m_k$ to delete all requests which followed the boundary of $\hole_b$.
        Note that if the \pActive{} module would become \pPassive{} instead in this case, then invariant (1) would be violated: by \cref{def:request-path}, $P$ is contained in the boundary of the target hole $\hole_m$, but since $\hole_b$ is distinct from $\hole_m$, they do not necessarily share the same boundary cells.
        Since $i < k$, the cleanup message must eventually arrive at $m_i$, after which the \pActive{} module is again side-adjacent to the last module of $P$.
        The cleanup message iteratively removes the last edge of $P$, such that the invariants carry over in each step.
        Otherwise, if $\hole_b$ contains the target cell, then the \pActive{} module becomes \pPassive{}.
        The boundary of $\hole_b$ contains all stored requests such that (1) holds.
        After becoming \pPassive{}, the module simulates receiving a request from $m_k$, i.e., it becomes $m_{k+1}$ in the next round, and (2) holds.
        Finally, (3) holds since modules delete all stored requests upon becoming \pActive{}.
        Once an \pActive{} module reaches the target cell and there are no cleanup messages currently being forwarded, the request path and passive segment are empty, and the invariants again hold trivially, as shown at the start of this proof.        

        This concludes that all invariants are maintained throughout execution.
    \end{proof}
}{
    \begin{proof}[Proof sketch]
        The invariants are shown by induction over the rounds of the algorithm.
        Before any request is initiated, the properties hold trivially.
        In each round, we analyze how the request path $P$ changes.
        Forwarding a request appends a single edge to $P$; a simple case distinction shows that all invariants are preserved.
        When a \pPassive{} module becomes \pActive{}, it is always the last module in $P$, and its request is deleted, which removes only the last edge of $P$.
        The \pActive{} module then moves by following $P$ in reverse.
        As long as it stays within the same pseudo-hole, by \cref{lem:move-in-phole} it can remain side-adjacent to last module in $P$ while iteratively deleting the request corresponding to the last edge in $P$. 
        Whenever the last edge of $P$ is deleted, the invariants directly carry over from the previous round. 

        Special care is needed when the \pActive{} module is blocked by a critical pair.
        The difficult case is where $m_{i+1} \neq m_{i-1}$ and $m_{i+1} \in N(m)$ (refer to \cref{subfig:proof-sketch-1} for the notation used).
        After $m$ becomes \pPassive{}, hole $\hole$ contains the target cell.
        We can show that $\hole \neq \hole_m$ by considering two closed simple curves $\mathcal{C}_1, \mathcal{C}_2$ that go once through both sides that $m_{i+1}$ shares with hole $\hole{}$ and $\hole_m$, respectively.
        One of the two curves must exist if $\hole = \hole_m$: $\mathcal{C}_1$ encloses all successors of $m_{i+1}$ but not $m_k$, and $\mathcal{C}_2$ encloses $m_k$ but no successors of $m_{i+1}$.
        Both contradict that $m_k$ is a successor of $m_{i+1}$. 
        Thus it holds that $\hole = \hole_b$.
    \end{proof}
    
    \begin{figure}[t]
        \centering
            \begin{subfigure}[b]{0.333\linewidth}
                \centering%
                \includegraphics[width=\linewidth,page=1]{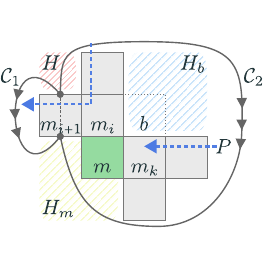}%
                \subcaption{}
                \label{subfig:proof-sketch-1}
            \end{subfigure}%
            \begin{subfigure}[b]{0.333\linewidth}
                \centering%
                \includegraphics[width=\linewidth,page=2]{proof-sketch}%
                \subcaption{}
                \label{subfig:proof-sketch-2}
            \end{subfigure}%
            \begin{subfigure}[b]{0.333\linewidth}
                \centering%
                \includegraphics[width=\linewidth,page=3]{proof-sketch}%
                \subcaption{}
                \label{subfig:proof-sketch-3}
            \end{subfigure}%
        \caption{(a) An \pActive{} module $m$ side-adjacent to a critical pair $(m_k, m_i)$ with bridge cells $m$ and $b$. Module $m_i$ is a predecessor of $m_k$ in $P$, and its direct successor $m_{i+1}$ is contained in $N(m)$. $\mathcal{C}_1$ and $\mathcal{C}_2$ are entirely within $\hole$ and contradict that $\hole = \hole_m$. (b) The simple polygon $\outline{}$ does not enclose the target hole; at least four vertices have exterior angle $\frac{\pi}{2}$. (c) The polygon $O$ encloses the target hole; at least two vertices have exterior angle $-\frac{\pi}{2}$ or there is a side of length at least two.}
        \label{fig:proof-sketch}
    \end{figure}
}

We call the request path \emph{well-formed} if it satisfies the three properties from \cref{lem:invariants}.
Assuming that the request path is well-formed, we can show that whenever a passive segment cycles back on itself, then that cycle contains a conditionally removable module.

\begin{lemma}
    \label{lem:request-loop-1}
    Let $P = (m_1, m_2, \ldots, m_k)$ be a well-formed request path with non-empty passive segment $S = (m_i, \ldots, m_k)$.
    Suppose the next anti-clockwise module along the target boundary after $m_k$ is some module $m_j$ with $i \leq j \leq k-2$.
    Then at least one of the modules $m_j, m_{j+1}, \ldots, m_k$ is conditionally removable.
\end{lemma}

\iftoggle{fullversion}{
    \begin{proof}
        Since $j \leq k-2$, there is at least one module between $m_j$ and $m_k$ in $S$.
        Note that otherwise $m_j$ would be the only side-adjacent neighbor of $m_k$, as it is both its direct predecessor in $P$, and the next anti-clockwise module along the target boundary, in which case $m_k$ would be simply removable.
        Since $S$ is simple (as we assume $P$ to be well formed), it follows that $Z \coloneqq (m_j, m_{j+1}, ..., m_k)$ is a simple cycle.

        First, assume that there is a critical pair $(c_1,c_2)$ with $c_1,c_2 \in Z$. 
        Since $Z$ is a cycle, each module in $Z$ has at least two side-adjacent neighbors within $Z$.
        By definition of a critical pair, both $c_1$ and $c_2$ have at most two side-adjacent neighbors.
        It follows directly that both $c_1$ and $c_2$ are conditionally removable.

        \begin{figure}[tbp]
            \centering
            \begin{subfigure}[b]{0.333\linewidth}
                \centering%
                \includegraphics[width=\linewidth,page=1]{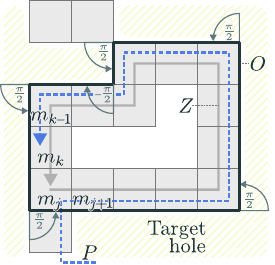}%
                \subcaption{}
                \label{subfig:request-loop-1}
            \end{subfigure}%
            \begin{subfigure}[b]{0.333\linewidth}
                \centering%
                \includegraphics[width=\linewidth,page=2]{request-loop}%
                \subcaption{}
                \label{subfig:request-loop-2}
            \end{subfigure}%
            \begin{subfigure}[b]{0.333\linewidth}
                \centering%
                \includegraphics[width=\linewidth,page=3]{request-loop}%
                \subcaption{}
                \label{subfig:request-loop-3}
            \end{subfigure}%
            \caption{Illustration of the two cases from \cref{lem:request-loop-1} and the terminology used in \cref{lem:request-loop-2}. (a) $\outline{}$ does not enclose the target hole; at least four vertices have exterior angle $\frac{\pi}{2}$. (b) $\outline{}$ encloses the target hole; at least two vertices have exterior angle $-\frac{\pi}{2}$ or there is a side of length at least two. (c) The C-shape in the common neighborhood of $m$ and its side-adjacent empty cell $e$.}
            \label{fig:request-loop}
        \end{figure}

        Second, assume the contrary.
        By \cref{def:request-path}, $Z \subseteq P$ is fully contained in the target boundary.
        Consider removing all modules side-adjacent to $m_j$ except for those contained in $Z$, i.e., except for $m_{j+1}$ and $m_k$.
        The remaining set of modules decomposes into at least two side-connected components, one of which contains $Z$ as its boundary (see \cref{subfig:request-loop-1,subfig:request-loop-2}).
        Let $\outline{}$ be the point set obtained by intersecting the union of all cells occupied by $Z$ with the target hole (including the removed neighbors of $m_j$).
        Since we assume that $Z$ contains no critical pair, $\outline{}$ is a simple polygon whose vertices are corners of the cells occupied by $Z$.

        There are two cases to distinguish.
        First, assume $\outline{}$ does not enclose the target hole (see \cref{subfig:request-loop-1}).
        Since we operate on the square grid, the exterior angle at each polygon vertex, defined as the angle between one side and the extension of the adjacent side, is $\pm \frac{\pi}{2}$.
        The sum of exterior angles of any simple polygon is $2\pi$.
        Thus, there are at least four vertices with exterior angle $\frac{\pi}{2}$, each corresponding to a cell sharing exactly two consecutive sides with the target hole.
        At most one of these cells can be occupied by $m_j$, so there are at least three conditionally removable modules in $Z$.

        Next, assume $\outline{}$ encloses the target hole (see \cref{subfig:request-loop-2}).
        The target cell is always adjacent to some non-passive module.
        Since $Z$ contains only \pPassive{} modules (as it lies entirely within the passive segment $S$), $\outline{}$ must enclose one additional cell besides the target cell.
        This implies that the width or height of $\outline{}$ is at least two cells.
        Thus, $\outline{}$ has two vertices with exterior angle $-\frac{\pi}{2}$ or a side of length at least two.
        Vertices with exterior angle $-\frac{\pi}{2}$ correspond to cells in $Z$ sharing two consecutive sides with the target hole (with $Z$ now lying on the outside of the polygon).
        At most one of these cells can be occupied by $m_j$, so the other is conditionally removable.
        If $\outline{}$ contains a side of length at least two, then there are consecutive $a, b \in Z$ both side-adjacent to the target hole. 
        Let $c$ and $d$ be the cells side-adjacent to $a$ and $b$, respectively, that are not in the target hole and not within $Z$ (see \cref{subfig:request-loop-2}).
        If $c$ is empty, then $a$ is conditionally removable.
        If $d$ is empty, then $b$ is conditionally removable.
        If neither is empty, both $a$ and $b$ are conditionally removable.

        In all cases, $Z$ contains a conditionally removable module, concluding the lemma.
    \end{proof}
}{
    \begin{proof}[Proof sketch]
        Let $Z \coloneqq (m_j, m_{j+1}, \ldots, m_k)$ be the portion of the request path closing the loop.
        Since $Z$ lies entirely within the passive segment $S$, it forms a simple cycle by \cref{lem:invariants}.

        If $Z$ contains a critical pair module, then that module is conditionally removable.
        Otherwise, we analyze the geometry of $Z$ on the grid.
        Consider the region outlined by $Z$, including the cells in $Z$ if $Z$ does not enclose the target hole (see \cref{subfig:proof-sketch-2}), and excluding them otherwise (see \cref{subfig:proof-sketch-3}).
        The boundary of this region forms a simple polygon $\outline{}$.

        If $Z$ does not enclose the target hole, then by standard angle-sum arguments for simple polygons, $\outline{}$ has at least four vertices with exterior angle $\frac{\pi}{2}$, at least one of which corresponds to a conditionally removable module in $Z$.
        If $Z$ does enclose the target hole, then it must also enclose the \pHead{} module, implying the existence of two vertices with exterior angle $-\frac{\pi}{2}$ or a side of length at least two, both of which imply a conditionally removable module in $Z$.
    \end{proof}
}

In the previous proof, we relied on the fact that the simple cycle~$Z$ consisted entirely of \pPassive{} modules. 
If the search for a conditionally removable module is triggered when the next forwarding direction of a request is not \pPassive{}, this assumption no longer holds. 
To handle this situation, we require the following property:  
A finite set of cells $V \subset C$ is \emph{weakly convex} if for any pair of cells $c_1, c_2 \in V$, there exists a shortest path from $c_1$ to $c_2$ (with respect to Manhattan distance) that lies entirely within~$V$.

If the set of all non-\pPassive{} modules is weakly convex, we can argue that a conditionally removable module exists if the target hole has size at least two.
However, this does not hold in general: there are configurations with a target hole of size one in which no module is conditionally removable.
This is precisely the edge case discussed in \iftoggle{fullversion}{\cref{par:special-case}}{Appendix~\ref{par:special-case}}.

\iftoggle{fullversion}{
    In the following, we show that the rhombus, together with the neighborhood of the next cell to be filled, is always weakly convex.
    This is sufficient for our algorithm to progress even in the edge case.
}{
    It can be shown (see Appendix~\ref{sec:appendix-proofs}) that the rhombus is weakly convex, and that the rhombus including the neighborhood of the next cell to be filled is weakly convex as well.
    This is sufficient for our algorithm to progress even in the edge case.
}

\iftoggle{fullversion}{
    \begin{lemma}
        \label{lem:extended-rhombus-weakly-convex}
        For any $n \geq 0$, the set of cells $\{v_0, \ldots v_{n}\} \cup N(v_n)$ is weakly convex, where $v_0, ..., v_n$ are the cells of a rhombus of size $n+1$.
    \end{lemma}

    \begin{proof}
        We first prove by induction that $\rhombus_n = \{v_0, \ldots, v_{n-1}\}$ is weakly convex for all $n \geq 0$.
        
        When $n = 0$, $\rhombus_0 = \emptyset$, which is trivially weakly convex.
        Let $n > 0$, let $c_1, c_2 \in \rhombus_n$ be arbitrary, and assume that $\rhombus_{n-1}$ is weakly convex.
        We distinguish three cases:
        \emph{Case 1:}
        $c_1 = c_2 = v_{n-1}$.
        A single cell is trivially weakly convex.
        \emph{Case 2:}
        $c_1, c_2 \in \rhombus_{n-1}$.
        By the induction hypothesis, $\rhombus_{n-1}$ is weakly convex.
        Hence there is a shortest path from $c_1$ to $c_2$ fully within $\rhombus_{n-1} \subset \rhombus_{n}$.
        \emph{Case 3:}
        $c_1 = v_{n-1}, c_2 \in \rhombus_{n-1}$.
        In the square grid, all shortest paths are monotone paths, e.g., they go east and south but never west and north.
        Assume that $c_1$ is extremal within its rhombus layer, w.l.o.g. northernmost.
        Then all cells in $\rhombus_{n-1}$ lie in the southeast or southwest quadrant relative to $c_1$.
        We obtain a shortest path from $c_1$ to $c_2$ by taking one step south, and afterwards take a southeast or southwest monotone path within $\rhombus_{n-1}$, which exists by the induction hypothesis.
        Now assume that $c_1$ is not extremal within its rhombus layer.
        W.l.o.g., assume that it lies in the northwest quadrant of $\rhombus_n$.
        Then there is no cell in the northwest quadrant relative to $c_1$, and $c_1$ has neighbors in $\rhombus_{n-1}$ to the east and south.
        If $c_2$ lies in the southwest quadrant, take one step south and follow a southwest monotone path within $\rhombus_{n-1}$.
        If $c_2$ lies in the southeast quadrant, take one step south and follow a southeast monotone path within $\rhombus_{n-1}$.
        If $c_2$ lies in the northeast quadrant, take one step east and follow a northeast monotone path within $\rhombus_{n-1}$.
        In all cases, we take a single step from $c_1$ to some neighbor in $\rhombus_{n-1}$, and afterwards follow a shortest path within $\rhombus_{n-1}$ maintaining monotonicity.
        Hence, the result is a shortest path entirely within $\rhombus_{n}$, which concludes that $\rhombus_{n}$ is weakly convex for all $n \geq 0$.

        It remains to show that $R_n \coloneq \rhombus_{n+1} \cup N(v_n)$ is weakly convex.
        Let $c_1, c_2 \in R_n$ be arbitrary. 
        We distinguish three cases:
        \emph{Case 1:}
        $c_1, c_2 \in \{v_n\} \cup N(v_n)$.
        $\{v_n\} \cup N(v_n)$ is a $3 \times 3$ square of cells. 
        It is easy to see that a full square is weakly convex.
        \emph{Case 2:}
        $c_1, c_2 \in \rhombus_{n+1}$.
        Since $\rhombus_{n+1}$ is weakly convex, there is a shortest path from $c_1$ to $c_2$ within $\rhombus_{n+1} \subset R_n$.
        \emph{Case 3:}
        $c_1 \in N(v_n) \setminus \rhombus_{n+1}, c_2 \in \rhombus_{n+1}$.
        If $v_n$ is extremal within its rhombus layer, w.l.o.g. northernmost, then no cell in $\rhombus_{n}$ lies north of any cell in $N(v_n) \setminus \rhombus_n$.
        All cells in $N(v_n) \setminus \rhombus_{n+1}$ have a neighbor to the south.
        Hence, we can move from $c_1$ into $\rhombus_{n+1}$ by both southwest and southeast monotone paths, and afterwards take a path within $\rhombus_{n+1}$ that follows the same monotonicity.
        If $v_n$ is not extremal, w.l.o.g. it lies in the northeast quadrant of $\rhombus_{n+1}$, then there is no cell in the northeast quadrant relative to $c_1$ that is contained in $\rhombus_{n+1}$.
        We again distinguish the quadrant in which $c_2$ lies relative to $c_1$.
        If $c_2$ lies in the northwest quadrant, move west from $c_1$ into $\rhombus_{n+1}$ and follow a northwest-monotone shortest path to $c_2$.
        If $c_2$ lies in the southeast quadrant, move south from $c_1$ into $\rhombus_{n+1}$ and follow a southeast-monotone shortest path to $c_2$.
        If $c_2$ lies in the southwest quadrant, then: 
        Move west from $c_1$ as long as $c_2$ is not strictly south or until entering $\rhombus_{n+1}$, then move south until $c_2$ is not strictly west or until entering $\rhombus_{n+1}$, and once inside $\rhombus_{n+1}$, follow a southwest-monotone path to $c_2$.

        In all cases, we construct a shortest path from $c_1$ to $c_2$ that lies entirely within $R_n$.
        Hence, $R_n = \{v_0, ..., v_{n-1}\} \cup \{v_n\} \cup N(v_n)$ is weakly convex.    
    \end{proof}
}{}

\begin{lemma}
    \label{lem:request-loop-2}
    Let $P = (m_1, m_2, \ldots, m_k)$ be a well-formed request path with non-empty passive segment $S = (m_i, \ldots, m_k)$.
    Suppose the next anti-clockwise module along the target boundary after $m_k$ is not \pPassive{}, and that $m_k$ is not simply removable.
    If no module in $S$ is conditionally removable, then the target hole has size one.
\end{lemma}

\iftoggle{fullversion}{
    \newcommand{\fixed}{\mathcal{F}}
    \begin{proof}
        The general structure of the proof is as follows.
        Let $D$ denote the set of all cells in $S$ that are side-adjacent to the target hole.
        Assuming that no module in $S$ is conditionally removable (see \cref{def:conditionally-removable}), we show the following three properties:
        (1) $D$ is non-empty, 
        (2) no critical pair $(c_1,c_2)$ exists with $c_1 \in D$ and $c_2 \in M$,
        and (3) for any consecutive $m_j, m_{j+1}$ in $S$, if $m_j \in D$, then $m_{j+1} \notin D$.
        We then consider the intersection of the neighborhood of any $m \in D$ with the neighborhood of some side-adjacent empty cell $e$.
        Using properties (1)--(3), we argue that modules within this intersection form a C-shape containing a module in $D$ that is distinct from $m$.
        The lemma then follows by considering the union of the C-shapes of two modules in $D$ together with property (2).
        We now proceed to prove each of the above claims in order.
        
        Let $m_l$ be the current position of the \pHead{} module within the rhombus $\rhombus{}$.
        Define $\fixed{}$ as the set of all \pHead{} and \pTail{} modules, i.e., the first $l+1$ cells of the rhombus.
        If there is a temporary \pHead{} module, i.e., we are in the special case handling described in \cref{par:special-case}, then include $m_{l+1} \cup N(m_{l+1})$ in $\fixed{}$.
        We have shown in \cref{lem:extended-rhombus-weakly-convex} that $\fixed{}$ is a weakly convex set.
        Denote $m'$ the next anti-clockwise along the target boundary after $m_k$, and let $P'$ be a shortest path from $m'$ to $m_{i-1}$ that is fully contained within $\fixed{}$ (which exists since $\fixed{}$ is weakly convex).
        Note that $P'$ and $S$ are disjoint since $S$ contains only \pPassive{} modules, and $S$ is simple by the assumption that $P$ is well-formed.
        It follows that the concatenation $Z = S \circ P'$ is a simple cycle.

        Claim (1): $D$ is non-empty.
        Assume by contradiction that $S$ consists only of the module $m_k$.
        Let $m' \in \fixed{}$ denote the next anti-clockwise module along the target boundary after $m_k$.
        Since every request is initiated by some module in $\fixed{}$, it follows that $m_{k-1} \in \fixed{}$.
        If $m' = m_{k-1}$, then $m_k$ has only one side-adjacent neighbor, contradicting the assumption that it is not simply removable.
        Hence, assume $m' \neq m_{k-1}$.
        If $m'$ and $m_{k-1}$ are not corner-adjacent, then the unique shortest path from $m'$ to $m_{k-1}$ must include $m_k \notin \fixed{}$, contradicting the weak convexity of $\fixed{}$.
        Therefore, assume that $m'$ and $m_{k-1}$ are corner-adjacent.
        By weak convexity, there exists a shortest path from $m'$ to $m_{k-1}$ that includes a module corner-adjacent to $m_k$, again contradicting that $m_k$ is not simply removable.
        We conclude that $S$ contains at least two modules.
        Since $S$ forms a contiguous path along the target boundary, at least one of these modules must be side-adjacent to the target hole.
        It follows that $D$ is non-empty.

        Claim (2): No critical pair $(c_1,c_2)$ exists with $c_1 \in D$ and $c_2 \in M$.
        Each module $c_1 \in D$ has a side-adjacent empty cell (a cell in the target hole), and two side-adjacent neighbors in $Z$.
        If the fourth cell side-adjacent to $c_1$ is empty, then $c_1$ would be conditionally removable, contradicting our assumption that no module in $S$ is conditionally removable.
        It follows that $c_1$ has exactly three side-adjacent neighbors and therefore cannot be part of a critical pair, since any critical pair module has at most two side-adjacent neighbors.
        This concludes the proof of the claim.

        Claim (3): For any consecutive $m_j, m_{j+1}$ in $S$, if $m_j \in D$, then $m_{j+1} \notin D$.
        Assume by contradiction that both $m_j$ and $m_{j+1}$ are in $D$.
        As shown in the previous claim's proof, each of them must have exactly three side-adjacent neighbors: two of which are their predecessor and successor in $Z$.
        The side-adjacent neighbors of $m_{j}$ are $m_{j-1}, m_{j+1}$ and some module $a \in N(m_{j+1})$.
        Similarly, the side-adjacent neighbors of $m_{j+1}$ are $m_{j}$, its successor in $Z$, which is either $m_{j+2} \in S$ if it exists, or $m' \in \fixed{}$ otherwise, and some module $b \in N(m_{j})$.
        Together, these neighbors form a path from $a$ to $m_{j+1}$ via $b$ that is fully contained in the neighborhood $N(m_j)$.
        This implies that $m_j$ is conditionally removable, which contradicts our assumption that no module in $S$ is conditionally removable.
        Hence, no two consecutive modules in $S$ are both contained in $D$. 

        Let $m \in D$ be arbitrary (which exists by Claim (1)), and let $e$ be the empty cell in the target hole that is side-adjacent to $m$.
        By definition of $D$, such a cell $e$ exists, and it is unique, since $m$ is not conditionally removable, i.e., it has exactly three side-adjacent modules.
        We continue to show that the common neighborhood of $m$ and $e$ has a C-shape and contains a module from $D$ that is distinct from $m$.
        Consider the labeling of cells in $N(m) \cap N(e)$ as illustrated in \cref{subfig:request-loop-3}.
        In words: let $a$ and $b$ denote the successor and predecessor of $m$ in the cycle $Z$, respectively.
        Let $c$ and $d$ be the cells such that $(m, a, c, e)$ and $(m, b, d, e)$ each form a 4-cycle.

        First, consider the case where $a, b \notin \fixed{}$.
        In this case, there is a module at both $c$ and $d$; otherwise, it holds that $a \in D$ or $b \in D$, which would imply that two consecutive modules in $S$ belong to $D$, contradicting Claim (3).
        The unique shortest path from $c$ to $d$ goes through the empty cell $e$.
        If both $c$ and $d$ were in $\fixed{}$, this would contradict the weak convexity of $\fixed{}$.
        Hence, at least one of $c$ or $d$ must be outside of $\fixed{}$.
        Since both are side-adjacent to $e$, at least one of them must be contained in $D$.

        Second, consider the case where exactly one of $a, b$ is contained in $\fixed{}$.
        W.l.o.g., assume $a \notin \fixed{}$ and $b \in \fixed{}$.
        As in the previous case, there must be a module in cell $c$; otherwise $a \in D$, contradicting Claim (3).
        Any shortest path from $c$ to $b$ must contain $m$ or $e$.
        Since $m, e \notin \fixed{}$, weak convexity implies that $c \notin \fixed{}$.
        If there is a module in cell $d$, then there is nothing more to show.
        So assume by contradiction that cell $d$ is empty.
        Let $f_1, f_2$ and $f_3$ be the cells in $N(e)$ side-adjacent to $c, e$ and $d$, respectively (see \cref{subfig:request-loop-3}).
        Any shortest path from a cell $f_i$ to $b \in \fixed{}$ must contain one of the cells $c, e_s$ or $d$, all of which lie outside of $\fixed{}$.
        Thus, by weak convexity, none of the $f_i$ can be in $\fixed{}$.
        Moreover, none of the $f_i$ can be empty:
        If $f_1$ were empty, then $c \in D$ would be conditionally removable, contradicting our assumption.
        If $f_2$ were empty, then $f_1 \in D$ and $f_2 \in D$ would be consecutive in $S$, contradicting Claim (3).
        If $f_3$ were empty, then $f_2 \in D$ would be conditionally removable, again a contradiction.
        Finally, if $d$ were empty, then $f_2 \in D$ and $f_3 \in D$ would be consecutive in $S$, contradicting Claim (3) once more.
        Hence, by contradiction, cell $d$ must be occupied.

        Third, the case $a, b \in \fixed{}$ cannot occur, as the shortest path from $a$ to $b$ contains $m \notin \fixed{}$.

        We have shown that for any $m \in D$, the modules in the intersection $N(m) \cap N(e)$ form a C-shape that includes a module $m^* \in D$ with $m^* \neq m$.  
        Together with Claim~(2), we now complete the proof of the lemma.

        Consider the union of the two intersections $N(m) \cap N(e)$ and $N(m^*) \cap N(e)$.  
        This union covers all cells in the neighborhood $N(e)$, except for a single cell that is corner-adjacent to $e$ (e.g., if $m^* = c$ in \cref{subfig:request-loop-2}, then the only cell not covered is $f_3$).
        By the weak convexity of $\fixed{}$, there can be no critical pair $(c_1, c_2)$ with $c_1, c_2 \in \fixed{}$.  
        Combined with Claim~(2), this implies that the target hole contains exactly one pseudo-hole.
        Consequently, the single remaining cell not contained in the union of the two C-shapes must be occupied as well. 

        We conclude that the target hole has size one.
    \end{proof}
}{}

The following lemma shows that each request successfully fills or nearly fills the target cell within a linear number of rounds, which is crucial for termination and runtime.

\begin{lemma}
    \label{lem:fill-target}
    Consider a round in which a module initiates a request for some target cell, and assume at least one \pPassive{} module exists.
    Then, within $\O(n)$ rounds, either the target cell is filled or the target hole consists of just the target cell.
\end{lemma}

\iftoggle{fullversion}{
    \begin{proof}
    Let $e$ be the target cell, and assume that the target hole has size larger than one, as otherwise there is nothing to show.
    Consider the first round after the request for cell $e$ is initiated in which the request is \emph{not} forwarded to the next module along the target boundary.
    Let $P = (m_1, m_2, \ldots, m_k)$ be the request path and $S$ be its passive segment in that round.
    Let $\fixed{}$ be defined as in the proof of \cref{lem:request-loop-2}.
    By weak convexity of $\fixed{}$, no \pPassive{} module is fully enclosed by modules in $\fixed{}$.
    Moreover, since we maintain that the set of all modules is side-connected (since we only ever move simply removable and conditionally removable modules), and the request path traverses the target boundary starting in $\fixed{}$, it must eventually contain a \pPassive{} module, i.e., $S$ is non-empty.
    By \cref{lem:invariants}, $m_k$ is \pPassive{}.
    This implies that in that round module $m_k$ does not forward the request because (1) it is simply removable, (2) the next anti-clockwise module along the target boundary after $m_k$ is contained in $S$, or (3) that module is contained in $\fixed{}$.
    In all three cases, some \pPassive{} module $m_i$ with $i \leq k$ becomes \pActive{}.
        
    Let $\phole_i$ be the pseudo-hole containing cell $m_i$ after removing its module, and let $\phole_0$ be the pseudo-hole containing the target cell.
    Note that both simple and conditional removability require $m_i$ to be side-adjacent to the target hole.

    First, consider the case where $\phole_i \neq \phole_0$. 
    Consider a shortest corner-path from $m_i$ to target cell $e$ entirely within the target hole $\hole$.
    Let $b_1$ be the last cell on that path within $\phole_i$, and $b_2$ the first cell on that path \emph{not} within $\phole_i$.
    Note that these two cells are uniquely defined, since for any adjacent pseudo-hole, $\phole_i$ contains exactly one bridge cell.
    Let $(c_1, c_2)$ be the critical pair adjacent to the bridge cells $b_1, b_2$.
    At least one of $c_1, c_2$ must be contained on the request path.
    Let $j$ be the maximum index w.r.t. $P$ for which $m_j = c_1$ or $m_j = c_2$.
    Since $m_i$ and $b_1$ are in the same pseudo-hole, by \cref{lem:move-in-phole}, the \pActive{} module can move from $m_i$ to $b_1$ remaining side-adjacent to the boundary of $\phole_i$.
    Moving and deleting all requests until it is side-adjacent to $m_j$ requires $\O(k-j)$ rounds.
    Let $\mathcal{S}$ be the number of square sides that are shared between an empty cell in the target hole and an occupied cell in its boundary.
    Note that while an occupied cell may belong to the boundaries of multiple holes, each such square side can be uniquely assigned to the boundary of a single hole.
    After the \pActive{} module becomes \pPassive{}, the target hole decomposes into precisely two holes, one of which contains the now empty cell $m_i$.
    Before $m_i$ became \pActive{}, all nodes on $P$ with index larger than $j$ are contained in the boundary of that hole.
    Hence, it takes $\O(k-j)$ rounds to decrease $\mathcal{S}$ by $\Omega(k-j)$.
    Since $\mathcal{S} = \O(n)$ initially, it takes $\O(n)$ rounds until $\mathcal{S} \leq 4$, in which case either the target hole has size one, or $\phole_i = \phole_0$.

    Now consider the case where $\phole_i = \phole_0$.
    In this case, by \cref{lem:move-in-phole}, the \pActive{} module can directly move into the target cell.
    It takes $\O(2k - i)$ rounds until $m_i$ is \pActive{} and another $\O(i)$ rounds until it arrives at the target cell.
    Hence, it takes $\O(n)$ in total.
\end{proof}
}{
    \begin{proof}[Proof sketch]
        One can show that whenever \pPassive{} modules exist, eventually some module becomes \pActive{}.
        If the \pActive{} module $m$ lies in the same pseudo-hole as the target cell, it can reach the target cell in $O(n)$ rounds. 
        Otherwise, it becomes \pPassive{} at a bridge cell, effectively disconnecting the target hole. 
        This reduces the number $\mathcal{S}$ of sides shared between the target hole and its boundary by an amount proportional to the length of the request path segment erased before $m$ becomes \pPassive{}. 
        Since $\mathcal{S} = O(n)$ initially and $\mathcal{S}$ decreases proportional to the number of rounds it takes for $m$ to become \pActive{}, move, and then become \pPassive{} again, after $O(n)$ rounds an \pActive{} module is in the target pseudo-hole, or $\mathcal{S} \leq 4$, in which case the target hole has size one.
    \end{proof}
}

The following corollary captures how the \pHead{} phase progresses through the configuration.

\begin{corollary} 
    \label{cor:head-progress} 
    Let $m$ be a module that becomes the \pHead{} at some round during the execution of the algorithm.
    After $\O(n)$ rounds, either $m$ becomes a \pTail{} and another module takes over as the new \pHead{}, or $m$ terminates.
\end{corollary}

\iftoggle{fullversion}{
    This follows directly from \cref{lem:fill-target}, together with the handling of the edge case as discussed in \cref{par:special-case}.
    In that case, the target hole has size one and no module in its boundary is conditionally removable.
    To resolve this, a request for a different cell outside the target hole is sent.
    Regardless of whether this second cell is filled or becomes a new hole of size one, we have shown that the target boundary subsequently contains a conditionally removable module.
    Hence, at most two consecutive requests are necessary to fill the target cell, both of which take $\O(n)$ rounds.
}{
    This follows from \cref{lem:fill-target}, together with the handling of the edge case discussed in Appendix~\ref{par:special-case}.
    In that edge case, the target hole has size one and no module in its boundary is conditionally removable.
    To resolve this, a request for a different cell outside the target hole is sent.
    Regardless of whether this second cell is filled or becomes a new hole of size one, we have shown that the target boundary subsequently contains a conditionally removable module.
    Hence, two consecutive requests suffice to fill the target cell, both of which take $\O(n)$ rounds.
}
If no \pPassive{} modules remain to take over the \pHead{} role, the request traverses the entire target boundary in $O(n)$ rounds and the \pHead{} terminates.

\begin{theorem}
    \label{thm:conclusion}
    Let $(M_0, s, l)$ be an arbitrary, initially side-connected configuration of $n$ modules with a leader module at cell $l \in M_0$, and all modules except for the leader being in the same state.
    Let $\rhombus = (v_0, v_1, \ldots, v_{n-1})$ be a rhombus centered at cell $v_0 = l$.
    After $\O(n^2)$ rounds of the rhombus formation algorithm, all cells in $\rhombus{}$ are occupied by a terminated module, and side-connectivity is preserved in each of those rounds.
\end{theorem}

\begin{proof}
    By \cref{lem:retired-is-rhombus}, the rhombus is filled incrementally with each transition of some module into the \pHead{} phase.
    Since \pTail{} modules never change phase again, there are at most $n$ \pHead{} transitions in total.
    By \cref{cor:head-progress}, each transition happens within $\O(n)$ rounds of each other.
    Once there is no more \pPassive{} module, i.e., the \pHead{} is positioned at cell $v_{n-1}$, it takes another $\O(n)$ rounds until the \pHead{} terminates.
    It follows that the total number of rounds until the \pHead{} module terminates is $\O(n^2)$.
    Once this occurs, termination propagates through all remaining \pTail{} modules in at another $\O(n)$ rounds.
    Since \pTail{} modules never move, all rhombus cells eventually contain terminated modules.

    A module becomes \pActive{} only if it is simply removable (see \cref{def:simply-removable}) or conditionally removable as part of a simple cycle $Z$ within $M$ (see \cref{def:conditionally-removable}), both of which preserve connectivity.  
    Moreover, whenever $Z$ consists only of \pPassive{} modules, we have shown in \cref{lem:request-loop-1} that such a conditionally removable module is always found before leaving $Z$.
    It follows that all module movements maintain side-connectivity in every round.
\end{proof}

\section{Utilizing Parallelization}
\label{sec:parallelization}

\newcommand{\tree}{\ensuremath\mathcal{T}}

In the previous sections, we presented and analyzed a distributed rhombus formation algorithm in which modules fill positions within the rhombus sequentially in $\mathcal{O}(n^2)$ rounds.
\iftoggle{fullversion}{
    No sequential algorithm can be faster in the worst case, as illustrated in the following:

    Consider a rhombus consisting of $j + 1$ layers for some $j \geq 0$.
    Layer $L_0$ contains a single cell, and each layer $L_i$ for $1 \leq i \leq j$ contains exactly $4i$ cells.
    Therefore, the total number of rhombus cells is $1 + \sum_{i=1}^j 4i = 2j^2 + 2j + 1$.
    Now consider a straight line of $n = 2j^2 + 2j + 1$ cells that originates at the center of the rhombus.
    Since the rhombus contains $j + 1$ layers, exactly $2j^2 + j$ cells of the line lie outside the rhombus.
    It follows that at least
    $\sum_{i = 1}^{2j^2 + j} i = \Omega(j^4)$
    rounds are required to move all outer modules into the rhombus sequentially.
    Since $n = \Omega(j^2)$, this implies a lower bound of $\Omega(n^2)$ rounds.
}{
    It is easy to show that no sequential algorithm can be faster in the worst case, not even centralized algorithms.
}

To improve upon the runtime of the sequential algorithm, we must enable modules to move in parallel.
However, preserving connectivity during parallel movement is non-trivial.
While the property of simple removability can be detected locally, identifying conditionally removable modules requires extensive communication.
Additionally, not all of them can move at the same time without disconnecting the configuration. 
Coordinating mutual exclusive movement would require some form of leader election, which in itself takes $\Omega(n^2)$ rounds under constant memory constraints.
Instead, we follow a different strategy:

We begin with a preprocessing phase in which a spanning tree $\tree{}$, rooted at the initial leader module, is constructed. 
This can be done via a distributed breadth-first search in $\mathcal{O}(n)$ rounds, using only constant memory per module.
As long as $\tree{}$ remains a valid spanning tree, any leaf in $\tree{}$ can move freely without violating connectivity. 
Our parallel algorithm utilizes this by letting leaves move clockwise along the boundary of holes whenever the movement reduces their depth in $\tree{}$.
Notably, we do not compute or store depth values.
Instead, each module maintains a pointer to its parent (i.e., the neighbor closer to the root), and movement decisions are made locally based on the parent pointers within a module’s neighborhood.

\begin{figure}[tbp]
    \centering
    \begin{subfigure}[b]{0.333\linewidth}
        \centering%
        \includegraphics[width=\linewidth,page=1]{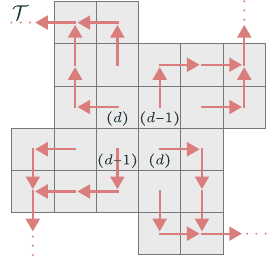}%
        \subcaption{}
        \label{subfig:locked-log}
    \end{subfigure}%
    \hspace{2cm} 
    \begin{subfigure}[b]{0.333\linewidth}
        \centering%
        \includegraphics[width=\linewidth,page=2]{locked-configs}%
        \subcaption{}
        \label{subfig:locked-simple}
    \end{subfigure}%
    \label{fig:locked-configs}
    \caption{(a) A configuration that remains locked even when modules have logarithmic memory. The configuration could be unlocked if any of the four central leaves (labeled by their depth in $\tree{}$) changed its parent pointer. However, all modules have minimum depth. (b) The movement of $m$ is blocked by $m'$ and vice versa. Both $m$ sliding to $b$ and setting its parent to $a$, or $m'$ sliding to $b'$ and settings its parent to $a'$, would unlock the configuration. However, neither module can locally determine which of the two moves improves their depth. }
\end{figure}

Following this strategy, there are two cases in which configurations may still become locked, meaning no viable move exists for any leaf.
In the first case, no leaf is side-adjacent to any hole.
Even if modules have logarithmic memory to compute and store their depths, such configurations cannot be resolved entirely locally (e.g., see \cref{subfig:locked-log}).
In the second case, no leaf can locally determine whether a move decreases its depth in $\tree{}$ (e.g., see \cref{subfig:locked-simple}).

To ensure that progress is eventually made, we run our sequential algorithm on top of the spanning tree strategy.
The sequential algorithm can also take advantage of the additional information provided by the spanning tree $\tree{}$.
In particular, whenever a \pPassive{} module checks whether it is simply removable, it also checks whether it is a leaf in $\tree{}$.
In both cases, its removal does not violate connectivity, and the module becomes \pActive{}.
However, to preserve the correctness of the sequential algorithm, the following issues must be addressed:

\begin{enumerate}
    \item A request or activation message can activate modules that are not necessarily leaves in $\tree{}$.  
    In this case, we locally reconfigure parent pointers to restore $\tree{}$ as a valid spanning tree.

    \item The movement of an \pActive{} module $m$ can be blocked by a leaf that arrives after $m$ becomes \pActive{}.

    \item Leaves can fill positions in the target hole side-adjacent to the request path.  
    In this case, modules on the request path are not necessarily contained within the target boundary.
\end{enumerate}

\subsection{Maintaining the Spanning Tree}
\label{subsec:parallel-issue-1}

In this subsection, we address the case where a request or activation message reaches and activates a module that is not a leaf in the spanning tree $\tree{}$.
If we consider arbitrary spanning trees over all modules, then removing a simply removable module (see \cref{def:simply-removable}) or a conditionally removable module (see \cref{def:conditionally-removable}) may locally break the tree structure in a way that cannot be repaired without global coordination. 
In particular, attempting to reroot the children of the removed module may disconnect subtrees or introduce cycles.

To preserve the tree structure, we must reroot all children of the removed module $m$ to neighbors that are not descendants of $m$ in $\tree{}$.
We say that a module $v$ is a \emph{descendant} of a module $u$ if the unique path from $v$ to the root of $\tree{}$ contains $u$; conversely, $u$ is then called an \emph{ancestor} of $v$.
However, since determining whether a neighbor is a descendant may require global information, we impose additional structure on $\tree{}$.

Specifically, we maintain the invariant that the subtree of $\tree{}$ induced by all \pHead{} and \pTail{} modules is a shortest-path tree rooted at the center of the rhombus. 
To preserve this invariant throughout the execution, whenever the \pHead{} role is delegated to a new module, the new \pHead{} sets its parent pointer toward the center of the rhombus; this direction can be inferred from the relative position of the previous \pHead{} module.

\begin{figure}[tbp]
    \centering
    \begin{subfigure}[b]{0.4271\linewidth}
        \centering%
        \includegraphics[width=\linewidth,page=1]{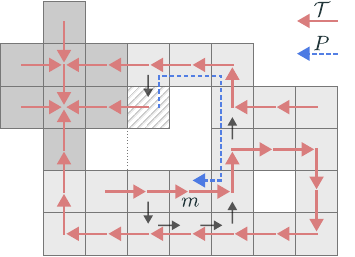}%
        \subcaption{}
        \label{subfig:rerooting-1}
    \end{subfigure}%
    \hspace{1cm}
    \begin{subfigure}[b]{0.4271\linewidth}
        \centering%
        \includegraphics[width=\linewidth,page=2]{rerooting.pdf}%
        \subcaption{}
        \label{subfig:rerooting-2}
    \end{subfigure}%
    \caption{Illustration of how tree edges change during forwarding of requests and activation messages, and when a module becomes \pActive{}. In (a), $m$ becomes \pActive{} upon receiving a request; in (b), $m_{k-2}$ becomes \pActive{} upon receiving an activation message initiated by $m_k$. Solid red arrows show tree edges before the request; solid black (smaller) arrows show redirected edges.}
    \label{fig:rerooting}
\end{figure}

\paragraph*{Rerooting Children of Simply Removable Modules}

As a request is forwarded along the target boundary, \pPassive{} modules that receive the request set their parent pointer in the opposite direction of the request path.
Denote $\tree{}'$ as the subgraph of $\tree{}$ induced by the set of all \pHead{} and \pTail{} modules, together with any module on the request path.
Within the subset of non-\pPassive{} modules, $\tree{}'$ is a shortest-path tree. 
Since the passive segment is simple by \cref{lem:invariants}, it follows that $\tree{}'$ is itself a tree.

What makes this structure useful is that the ancestors of the last module on the request path can now be identified locally, since they all lie on the request path.
All modules side-adjacent to some simply removable module $m$ also lie in the same connected component within the neighborhood of $m$.
Consequently, when $m$ becomes \pActive{}, its children in $\tree{}$ can be locally rerooted to one of its ancestors.
See \cref{subfig:rerooting-1} for an illustration of this procedure.

\paragraph*{Rerooting Children of Conditionally Removable Modules}

In contrast to simply removable modules, modules side-adjacent to a conditionally removable module $m$ lie in two distinct connected components within $N(m)$.
One of these components contains the predecessor of $m$ on the request path, which, as argued in the previous paragraph, is not a descendant of $m$.
To safely reroot all children of $m$, we must identify non-descendants of $m$ within the other component.
This is achieved using the following strategy:

Let $m_k$ be the last module on the request path, $m_{k+1}$ be the next module along the target boundary in anti-clockwise order, and $m_{k-1}$ be the predecessor of $m_k$ on the request path.
Suppose that $m_k$ would normally forward a request to $m_{k+1}$, but $m_{k+1}$ is either a \pHead{} or \pTail{} module, or a \pPassive{} module already on the request path.
In either case, $m_{k+1}$ is not a descendant of $m_k$.

Now suppose that $m_k$ is not conditionally removable.
Following our algorithm, it deletes its stored request and sends an activation message to $m_{k-1}$. 
In addition, it now sets its parent pointer to $m_{k+1}$.
As a result, when $m_{k-1}$ checks in the next round whether it is conditionally removable, $m_{k}$ is no longer its descendant.
This process repeats until eventually a conditionally removable module becomes \pActive{}.
At that point, we reroot all its children to neighboring non-descendants, which we can now identify within all connected components in its neighborhood.
See \cref{subfig:rerooting-2} for an illustration of this procedure.

\subsection{Moving Through Leaves}
\label{subsec:parallel-issue-2}

In this subsection, we consider all cases in which the movement of an \pActive{} module is blocked by a leaf module that does not belong to the request path.
For each such case, we propose a strategy that allows the \pActive{} module to make progress via local state swaps: rather than moving into the blocked position, the \pActive{} module delegates its role to the blocking leaf module.

\begin{figure}[tbp]
    \centering
    \begin{subfigure}[b]{0.333\linewidth}
        \centering%
        \includegraphics[width=\linewidth,page=1]{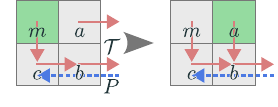}%
        \subcaption{}
        \label{subfig:parallel-issue-2-1}
    \end{subfigure}%
    \hfill
    \begin{subfigure}[b]{0.333\linewidth}
        \centering%
        \includegraphics[width=\linewidth,page=2]{parallel-issue-2}%
        \subcaption{}
        \label{subfig:parallel-issue-2-2}
    \end{subfigure}%
    \hfill
    \begin{subfigure}[b]{0.333\linewidth}
        \centering%
        \includegraphics[width=\linewidth,page=3]{parallel-issue-2}%
        \subcaption{}
        \label{subfig:parallel-issue-2-3}
    \end{subfigure}%
    \caption{Illustration of how an \pActive{} module at $m$ can proceed when blocked by a leaf not on the request path $P$. Cases shown: (a) slide blocked at $a$; (b) convex transition blocked at $b$; (c) convex transition blocked at $a$ with $b$ empty. The case where $b$ is not empty is analogous to (b). All possible cases are shown except for variations in the parent pointer of the blocking leaf.}
    \label{fig:parallel-issue-2}
\end{figure}

First, consider a slide from cell $m$ to cell $a$, where $m$ is the \pActive{} module, and let $(m,a,b,c)$ form a clockwise 4-cycle in which $c$ is the last, and $b$ is the second last module on the request path.
If $a$ is a leaf of $\tree{}$ with $a \notin  P$, then $m$ becomes \pPassive{} and sets its parent pointer to $c$, $a$ becomes \pActive{} and sets its parent pointer to $b$, and $c$ deletes its stored request from $b$; see \cref{subfig:parallel-issue-2-1}.

Second, consider a convex transition from cell $m$ to cell $b$, and let $(m, a, b, c)$ be a clockwise 4-cycle in which $c$ is the last module on the request path.
If $b$ is a leaf module with $b \notin P$, we can apply the same strategy: $m$ becomes \pPassive{} while $b$ becomes \pActive{}, and both $m$ and $b$ set their parent pointers to $c$; see \cref{subfig:parallel-issue-2-2}.
If $b$ is empty and $a$ is a leaf module with $b \notin P$, then $m$ becomes \pPassive{} while $a$ slides to $b$ and becomes \pActive{}, and both set their parent pointers to $c$ again; see \cref{subfig:parallel-issue-2-3}.

We have considered a blocking leaf at every position that must be empty for a slide or convex transition to be viable. 
In all cases, the \pActive{} module can make progress while only manipulating the parent pointers of leaves, ensuring that the tree structure is preserved.

\subsection{Ignoring Unattached Leaves}
\label{subsec:parallel-issue-3}

To preserve the properties stated in \cref{lem:invariants}, we treat certain leaf modules analogously to how we treat the \pActive{} module, that is, as if they were an empty cell.
Specifically, we ignore any leaf module that is not part of the request path and whose parent pointer does not point to a module on the path.
This is justified for two reasons: first, we have already shown that the \pActive{} module can always move through such leaves; second, when considering the removal of a module $m$ on the request path, any leaf whose parent is not $m$ still has a path to the root within $\tree{}$.
If that path does not include $m$, then the removal of $m$ does not affect connectivity.
If it does include $m$, then it must also include a child and the parent of $m$.
Otherwise, that path must also contain a child and the parent of $m$, and the rerooting strategy described in \cref{subsec:parallel-issue-1} ensures that connectivity is preserved.

\begin{figure}[tbp]
    \centering
    \includegraphics[width=0.55\linewidth]{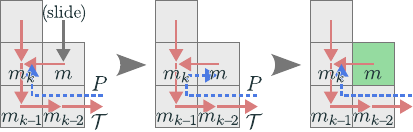}%
    \caption{A module $m$ sliding adjacent to $m_k$ in the same round that $m_k$ receives a request violates Invariant~1. Activating $m$ in the next round restores the invariant.}
    \label{fig:parallel-issue-3}
\end{figure}

The only remaining issue arises when a leaf module attaches itself to a module on the request path \emph{after} that module has forwarded its request.
To prevent this, we do not allow leaves to attach to the request path, and any leaf that is already attached to the request path is not allowed to move.
However, since a leaf may attach in the same round that a module receives a request and becomes part of the request path, some edge cases must be considered.

Consider an anti-clockwise 4-cycle $(m_k, m_{k-1}, m_{k-2}, m)$ where $m_k, m_{k-1}, m_{k-2} \in P$ are the last three modules on the request path, and $m \notin P$ is some leaf in $\tree{}$.
Note that the leaf module $m$ must have moved to its position after $m_{k-2}$ has forwarded the request to $m_{k-1}$, as otherwise it would have been forwarded directly to $m$.
As an example, in \cref{fig:parallel-issue-3}, $m$ slides into its position in the same round in which $m_k$ receives a request.
The case where $m$ moves via a convex transition is analogous.
In this round, $m_{k-1}$ may not be contained in the target boundary.
By \cref{def:request-path}, this implies that the request path ends at $m_{k-2}$, and Invariant 1 from \cref{lem:invariants} is temporarily violated.
However, in the next round $m$ receives the request from $m_k$ and becomes \pActive{} (since it is a leaf), and the invariant is immediately restored.

This edge case is unproblematic because the resulting configuration is the same as if module $m$ had first become \pActive{} and then moved into its current position.
Moreover, if the leaf $m$ does not form the above 4-cycle with three consecutive modules on the request path, then it would receive the request in exactly the same round, regardless of whether it had just moved into its position or was already present when the request was initiated.
Hence, all other cases are unproblematic as well.

\subsection{Theoretical Guarantees}

In summary, our parallelization strategy preserves correctness: we have addressed all issues that arise from parallel execution, including rerooting to maintain a spanning tree structure (\cref{subsec:parallel-issue-1}), interactions of the \pActive{} module with \pPassive{} leaves (\cref{subsec:parallel-issue-2}), and maintaining the invariants from \cref{lem:invariants} in presence of such leaves (\cref{subsec:parallel-issue-3}).
Combined with the correctness of the underlying sequential algorithm (\cref{thm:conclusion}), we conclude that our parallel extension reconfigures any initially connected configuration into a rhombus in $\O(n^2)$ rounds.
In the following section, we introduce two variants designed to further improve performance and analyze their runtime through simulations.
\section{Parallel Variants and Experimental Evaluation}
\label{sec:variants}

To further improve performance, we developed two variants, V1 and V2, that increase the degree of parallelism.
V1 allows leaf modules to perform state swaps similar to those of the \pActive{} module whenever their movement trajectories intersect.
This enables progress in configurations like the one in \cref{subfig:locked-simple} without waiting for requests.
Combined with the parallelization strategy from \cref{sec:parallelization}, V1 already yields a significant speedup over the purely sequential algorithm in most cases.
Importantly, the correctness of the algorithm does not depend on the specific movement choices of the leaves, as long as leaves are not allowed to attach to the request path.
However, since leaf modules are restricted to moves that improve their depth in $\tree{}$, certain configurations still become locked, in which case the runtime is asymptotically not better than the sequential algorithm.

V2 addresses this by relaxing the depth-improvement requirement.
Leaves move as in V1, with one addition: a leaf proceeds with an available move even if it cannot locally verify that the move improves its depth in $\tree{}$.
Such moves may temporarily increase a leaf's depth and delay its progress, but they can also turn its parent into a leaf, unlocking further parallel movement.
To prevent unproductive cycles in which leaves repeatedly step over each other, we do not allow leaves to attach themselves to another \pPassive{} leaf.
In our experiments, V2 consistently outperforms V1 and appears to achieve linear runtime on average.

\paragraph*{Simulation Setup and Evaluation}

We evaluated the performance of the sequential algorithm and both parallel variants (V1 and V2) on structured and random configurations using a custom-built simulator.
For each configuration type and size, we ran multiple independent simulations.

The structured input configurations considered are \emph{lines}, where modules are arranged in a straight horizontal segment; \emph{spirals}, generated by placing modules in a square spiral pattern; and \emph{chains}, constructed as cascaded sequences of critical pairs.
In all structured configurations, the initial leader is explicitly placed at a degree-one module.
The top row in \cref{fig:plots} shows the number of rounds required to reconfigure a line and a chain into a rhombus, for configurations of size 10 to 120 in increments of 10.
Spiral configurations are excluded from the plots since they follow the same trend as line configurations.
We also include the Earth Mover’s Distance (EMD) as a lower bound for any sequential solution, and the number of modules $n$ as a lower bound for V1 and V2.

\begin{figure}[htbp]
    \centering
    \includegraphics[width=0.5\linewidth,page=2]{spiral_algorithm_plots.pdf}%
    \includegraphics[width=0.5\linewidth,page=1]{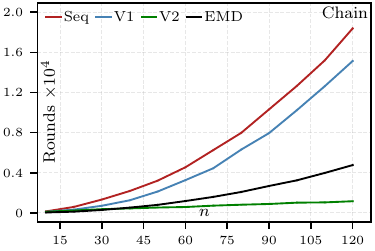}\\
    \includegraphics[width=0.5\linewidth,page=4]{spiral_algorithm_plots.pdf}%
    \includegraphics[width=0.5\linewidth,page=5]{spiral_algorithm_plots.pdf}\\
    \includegraphics[width=0.5\linewidth,page=10]{spiral_algorithm_plots.pdf}%
    \includegraphics[width=0.5\linewidth,page=11]{spiral_algorithm_plots.pdf}\\
    \includegraphics[width=0.5\linewidth,page=6]{spiral_algorithm_plots.pdf}%
    \includegraphics[width=0.5\linewidth,page=7]{spiral_algorithm_plots.pdf}%
    \caption{
    Comparison of the runtime of the sequential rhombus formation algorithm (Seq) with its two parallel variants (V1 and V2).
    The top row shows the runtime on line and chain configurations; the second row shows the runtime on random percentage configurations.
    In all but the top row, the right plot displays the same data but excludes the sequential runtime for better resolution.
    The third and fourth rows compare runtime against the Earth Mover's Distance (EMD) and the leader's eccentricity for random configurations (both percentage-based and Perlin noise-based) of fixed size $n = 200$.
    Scatter points indicate individual runs, while lines are smoothed using LOWESS (locally weighted regression over nearby points).
    }

    \label{fig:plots}
\end{figure}

As expected, the runtime of the sequential algorithm grows quadratically on both line and chain configurations, which are worst-case inputs in the sequential setting.
In contrast, the runtimes of V1 and V2 coincide on line and spiral configurations, where leaves are not blocked during reconfiguration.
However, for chain configurations, V1 yields only a constant-factor improvement over the sequential algorithm, while V2 appears to achieve linear runtime.

We used two different types of random input configurations:
\begin{enumerate}
    \item \emph{Percentage-based Configurations:}
    Starting from a filled rectangle, modules were removed uniformly at random while preserving connectivity, until a specified percentage remained.
    \item \emph{Growth-based Configurations (Perlin Noise):}  
    These configurations were grown from a single seed module.  
    In each step, all empty cells side-adjacent to existing modules were assigned a weight based on a fixed 2D Perlin noise map (a type of gradient noise).
    Among the 10 cells with the highest weights, one was selected uniformly at random and added to the configuration.  
    This method produces irregular, organic shapes with many holes.
\end{enumerate}

For each rectangle of dimension $(x,y)$ with $x$ ranging from 6 to 20 in steps of 2, and $y$ from $x$ to 20 in steps of 2, we generated 10 configurations with 10\%, 30\%, 50\%, 70\%, and 90\% of the original modules retained, resulting in a total of 1800 simulations.
The second row in \cref{fig:plots} shows the results of all those runs.

Both variants significantly improve runtime over the purely sequential algorithm and even outperform the theoretical lower bound given by the Earth Mover’s Distance (EMD).
This is expected, as the parallel algorithm allows multiple modules to move simultaneously, while the EMD assumes only one module moves per round.
Additionally, V2 consistently outperforms V1 across all simulations.
In conclusion, the speedup from increased parallelism in V2 outweighs the benefit of shorter paths in V1.

To analyze which configuration parameters influence runtime, we generated over 1000 random configurations, both percentage-based and growth-based, each with 200 modules.
For each configuration, we computed several metrics: its diameter, the number of holes and pseudo-holes, the leader module's eccentricity (i.e., the largest distance to any non-leader module), its closeness (i.e., the mean eccentricity of all modules), the EMD towards the rhombus, and the perimeter (i.e., the number of edges between empty and occupied cells).

Analyzing the plots, we find that only the leader's eccentricity, closeness, and EMD correlate with runtime. 
As eccentricity and closeness show nearly identical trends and other parameters appear uncorrelated, we include only the plots for eccentricity and EMD, shown in the third and fourth rows of \cref{fig:plots}, respectively.

The runtime increases roughly linearly with EMD, as expected: larger differences between the initial and target configurations require more total module movement. 
In contrast, the runtime grows sublinearly, approximately following a square-root curve, with the leader's eccentricity. 
This is likely because distant modules rarely move directly to the target cell.
Instead, they often fill critical pairs along the way, enabling closer modules to move first. 
As a result, the rhombus grows, reducing the distance that distant modules must traverse in later rounds.
In the EMD plot, V2 shows a smooth, steady linear runtime with minimal outliers, while V1 shows higher variance.
This suggests that locked leaves, as constructed manually in earlier examples, frequently occur even in random configurations, causing delays under V1 but not under V2.

\section{Conclusion}
\label{sec:conclusion}

We presented a distributed reconfiguration algorithm for the sliding square model that transforms side-connected initial configurations into a rhombus.
The algorithm operates under strict local communication and constant memory constraints, while maintaining connectivity throughout execution.
It runs sequentially with a worst-case runtime of~$\O(n^2)$, which is optimal for sequential solutions under these constraints.
While our focus was on the rhombus, the algorithm generalizes to any target shape that is sequentially constructible and weakly convex at each construction step, including, for example, lines and filled rectangles.

To improve runtime, we introduced two parallel variants V1 and V2.
Simulations show that both variants outperform the sequential algorithm and even surpass the theoretical lower bound for sequential runtimes given by the Earth Mover’s Distance, with V2 achieving the greater speedup.
We identified key factors influencing runtime: it scales linearly with EMD and follows a square-root curve with respect to the leader’s eccentricity and closeness.
V2, in particular, showed a smooth runtime curve with low variance, suggesting it can handle configurations that tend to become locked under V1.

Future work includes extending the algorithm to non-weakly convex shapes, exploring transformations from the rhombus to arbitrary configurations, and explore lower bounds to clarify the complexity of distributed reconfiguration in the sliding square model.

%
%
%
\bibliographystyle{splncs04}
\bibliography{bibliography}

\iftoggle{fullversion}{}{
\appendix
\section{Technical Details of the Algorithm}
\label{sec:appendix-technical-details}
\setcounter{lemma}{0}
\setcounter{figure}{0}
\renewcommand{\thefigure}{A.\arabic{figure}}

In the following, we describe the technical details omitted from \cref{sec:algorithm} to complete the description of our rhombus formation algorithm.

We first describe how a \pHead{} module determines the next target cell.
Let $\rhombus{} = (v_0,...,v_{n-1})$ be a rhombus of size $n$, and let $v_i$ be the cell in $\rhombus{}$ occupied by the \pHead{} module.
In our algorithm, we maintain as an invariant that $\{v_0, ..., v_{i-1}\}$ is precisely the set of all cells occupied by \pTail{} modules.
This ensures that the \pHead{} module can determine its position within the rhombus by observing the states of its neighbors: whether it is at the first cell ($v_i = v_0$), at an extremal cell (e.g., northernmost), on one of the layer's straight segments, or at its connector.
For example, if there are \pTail{} neighbors in direction $\N$ and $\E$ (e.g., \cref{subfig:example-1}), then $v_i$ lies on a straight segment in the rhombus' southwest quadrant, and the target cell is in the $\NW$ direction.

Next, we clarify how modules determine the direction in which messages are forwarded and whether they are side-adjacent to the target hole.
We assume that whenever a module $m$ receives a message, it knows which of its neighbors $src$ has sent the message.
This allows $m$ to compute both the direction of the target hole and the next forwarding direction along the target boundary.
Module $m$ scans its neighbors in anti-clockwise order, starting with module $src$.
The first empty cell following $src$ in this scan belongs to the target hole, and the first side-adjacent module following $src$ is the next module along the target boundary in anti-clockwise direction.
Similarly, when initiating a request, module $m$ performs a scan starting with the target cell instead of $src$.
When forwarding activation and cleanup messages, a module can directly read the forwarding direction from its stored requests.


\paragraph*{Holes of Size One}
\phantomsection
\label{par:special-case}
In this paragraph, we discuss the special case where an activation message does not reach a conditionally removable module.
Let $P = (m_1, ..., m_k)$ be the request path, $S = (m_i, ..., m_k)$ with $i > 1$ its passive segment, and let $m$ be a non-\pPassive{} module side-adjacent to $m_k$.
Recall that we obtain a simple cycle by concatenating $S$ with a shortest path from $m$ to $m_{i-1}$ within the set of all non-\pPassive{} modules.
In general, any activation message traversing $P$ in reverse reaches a conditionally removable module before it reaches $m_{i-1}$.
However, this may not the case if the target hole has size one, i.e., it consists of only the target cell.
There are configurations in which no \pPassive{} module on the target boundary is conditionally removable (e.g., see \cref{subfig-appendix:special-case1}).
In this case, the activation message fully traverses the request path in reverse until it reaches the \pHead{} module.

Although no \pPassive{} module is initially conditionally removable, we can still create a conditionally removable module using a second request.
Let $e$ be the target cell, and assume that the target hole has size one and that the target boundary contains no conditionally removable \pPassive{} module.
Let $Z$ be the simple cycle consisting of the eight modules in the target boundary in an anti-clockwise order starting at the \pHead{} module.
Let $m_1, m_2$ be the first two \pPassive{} modules in $Z$ that are side-adjacent to the target hole.
Since neither $m_1$ nor $m_2$ are conditionally removable, there exists two empty cells $e_1$ and $e_2$ corner-adjacent to $m_1$ and $m_2$, respectively, such that $e_1$ and $e_2$ are corner-adjacent to each other and thus lie in the same hole (see \cref{subfig-appendix:special-case1}).

Once the \pHead{} module receives an activation message, it instructs $m_1$ to initiate a new request for cell $e_1$.
At this point, $m_1$ temporarily acts as the \pHead{}, and all other modules within $Z$ temporarily act as \pTail{} modules.
By \cref{lem:fill-target}, this process leads to one of two outcomes (e.g., see \cref{subfig-appendix:special-case2,subfig-appendix:special-case3}):
(1) Cell $e_1$ is filled by some module.
Then $m_1$ becomes conditionally removable w.r.t. $Z$ and can slide into cell $e$.
(2) All cells in the neighborhood of $e_1$ are filled by modules while $e_1$ remains empty.
As a result, cell $e_2$ must be occupied by a module, which implies that $m_2$ is now conditionally removable w.r.t. $Z$ and can slide into cell $e$.
In either case, a module moves into cell $e$, after which all modules adjacent to $e$ revert to their original roles and the algorithm proceeds as usual.

In conclusion, for the purpose of the analysis, it suffices to focus on holes of size at least two:
Once the target hole $e$ is reduced to size one, at most one additional request for some intermediate cell $e_1 \neq e$ is initiated.
Whether this second request successfully fills the cell $e_1$ is irrelevant; in both cases, there is a conditionally removable module that can slide into $e$.

\begin{figure}[t]
    \centering
    \begin{subfigure}[b]{0.33\linewidth}
        \centering%
        \includegraphics[width=\linewidth,page=1]{special-case}%
        \subcaption{}
        \label{subfig-appendix:special-case1}
    \end{subfigure}%
    \begin{subfigure}[b]{0.33\linewidth}
        \centering%
        \includegraphics[width=\linewidth,page=2]{special-case}%
        \subcaption{}
        \label{subfig-appendix:special-case2}
    \end{subfigure}%
    \begin{subfigure}[b]{0.33\linewidth}
        \centering%
        \includegraphics[width=\linewidth,page=3]{special-case}%
        \subcaption{}
        \label{subfig-appendix:special-case3}
    \end{subfigure}%
    \caption{Configuration with a target hole of size one and no conditionally removable module in its boundary. Empty cells are marked by red dots; other cells may be empty or occupied. (a) Module $m_1$ initiates a request for $e_1$. (b) If $e_1$ is filled, $m_1$ becomes conditionally removable and slides into $e$. (c) If not, and the case repeats, $m_2$ becomes conditionally removable and slides into $e$.}
    \label{fig-appendix:special-case}
\end{figure}

\section{Deferred Proofs}
\label{sec:appendix-proofs}
\setcounter{lemma}{0}
\renewcommand{\thelemma}{B.\arabic{lemma}}
\renewcommand{\thefigure}{B.\arabic{figure}}

\begin{lemma}[Restatement of \cref{lem:invariants}]
    \label{lem-appendix:invariants}
    Let $P = (m_1, m_2, \ldots, m_k)$ be the request path.
    In every round during the execution of the rhombus formation algorithm, the following properties hold as invariants:
    \begin{enumerate}
        \item For every module $m$ that stores a request from some module $m'$, the directed edge $(m', m)$ is contained in $P$.
        \item The request path $P$ contains at most one passive segment $S = (m_i, ..., m_j)$, and if $S$ is non-empty, then $j = k$, i.e., the passive segment ends at the last module of $P$.
        \item The passive segment $S$ is simple, i.e., no module appears more than once in $S$.
    \end{enumerate}
\end{lemma}

\begin{proof}
    The invariants are shown by induction over the rounds of the algorithm.
    
    Initially, no request has been initiated, and no module stores any request.
    Hence, all properties hold trivially.
    Now suppose that module $m_1$ initiates a request, i.e., $P = (m_1)$ and the passive segment $S$ is empty.  
    Since no module has yet stored any request, Property~(1) holds.
    Since $S$ is empty, Properties~(2) and (3) hold trivially as well.

    Now consider an arbitrary round in which $P$ is not empty and assume that the invariants hold.
    We distinguish the possible actions that may occur in the current round:

    First, assume module $m_k$ forwards a request to some module $m_{k+1}$.
    A \pPassive{} module only forwards a request if $m_{k+1}$ is a \pPassive{} module and not contained in $P$.
    Hence, $P$ and $S$ extend by a \pPassive{} module $m_{k+1}$ previously not contained in $P$ and $S$, and all properties hold.
    Whenever a non-\pPassive{} module forwards a request, then by (2), the passive segment is empty in that round.
    Hence, afterwards, the passive segment remains empty or consists only of $m_{k+1}$ and the properties hold.

    Second, a \pPassive{} module $m$ becomes \pActive{} upon receiving a request or activation message.
    In the first case, it must be the last module in $P$ by invariant (2).
    In the second case, the activation message was received from its successor in $P$ w.r.t. the previous round.
    Only \pPassive{} modules send activation messages, and after doing so, they delete their last stored request.
    By invariant (3), the passive segment $S$ is simple.
    Therefore, the successor of $m$ w.r.t. the previous round has deleted its only stored request and is no longer in $P$.
    Hence, in both cases $m$ is the last module in $P$, i.e., $m = m_k$, and it deletes its request upon becoming \pActive{}.
    It follows, that the invariants carry over from the previous round, since only the last edge in $P$ is removed.

    \begin{figure}[t]
        \centering
        \includegraphics[width=0.5\linewidth,page=1]{active-critical-pair}%
        \caption{(Restated.) An \pActive{} module $m$ side-adjacent to a critical pair $(m_k, m_i)$ with bridge cells $m$ and $b$. Module $m_i$ is a predecessor of $m_k$ on the request path, and its direct successor $m_{i+1}$ is contained in $N(m)$. In this case $\hole_b = \hole$ contains the target cell.}
        \label{fig-appendix:active-critical-pair}
    \end{figure}

    Third, consider the deletion of requests by an \pActive{} module, either directly or by its cleanup messages.
    As shown above, only the last module in $P$ can become \pActive{} and afterwards deletes its stored request, which corresponds to the last edge in $P$.
    It follows that an \pActive{} module $m$ always begins its movement side-adjacent to the last module in $P$.  
    It then moves by following $P$ in reverse.
    As long as it stays within the same pseudo-hole, by \cref{lem:move-in-phole} it can remain side-adjacent to last module in $P$ while iteratively deleting the request corresponding to the last edge in $P$. 
    Whenever the last edge of $P$ is deleted, the invariants directly carry over from the previous round. 

    If the movement of $m$ is blocked by a critical pair, its behavior changes: the module either becomes \pPassive{} or initiates a cleanup message.  
    To prove that the invariants are still maintained in this case, we must show that the case distinction in \cref{par:critical-pair-cases} correctly determines which of the pseudo-holes adjacent to $m$ is part of the target hole.
    
    Let $(m_k, c)$ be the critical pair where $m_k$ is the last module on $P$.
    Let $b$ be the bridge cell adjacent to that pair, where $b \neq m$.
    Recall, that we treat the \pActive{} module as an empty cell w.r.t. our terminology.
    Hence, $m$ is the other bridge cell adjacent to the critical pair.
    Let $\hole{}_m$ be the hole that contains cell $m$, if $b$ would be occupied, and let $\hole_b$ be the hole that contains cell $b$, if $m$ would be occupied (see \cref{fig-appendix:active-critical-pair}).
    W.l.o.g., we assume that $c$ is in direction $\N$ of $m$.
    The other four cases are analogous up to rotation.

    First, consider the case where $c$ is contained in $P$.
    Consider the last appearance of $c$ in $P$, i.e., maximize the index $i$ such that $m_i = c$.
    Consider an anti-clockwise scan around $m_1$ starting from the target cell $e$ and note that $m_2$ is the first side-adjacent module in that scan after the target cell.
    It follows that in the anti-clockwise scan around $m_1$, any empty cell between $e$ and $m_2$ belongs to the target hole, and the last of those empty cell is adjacent to $m_2$.
    Continue this process with $m_2, m_3$ and so forth until module $m_i$.
    Note that its successor $m_{i+1}$ must be distinct from $m_k$, since $m_k$ and $m_i$ are not side-adjacent, i.e., $k > i +1$.
    If $m_{i+1} \notin N(m)$ or $m_{i+1} = m_{i-1}$, then the scan around $m_i$ from $m_{i-1}$ to $m_{i+1}$ contains an empty cell side-adjacent to $m$.
    It follows that $\hole_m$ contains the target cell.
    Consider the case where $m_{i+1} \in N(m)$ and $m_{i+1} \neq m_{i-1}$.
    Since $i > 1$ by assumption and $b$ is empty, $m_{i-1}$ must lie in direction $\N$ of $m_i$.
    In this case, the only empty cell in the scan around $m_i$ from $m_{i-1}$ to $m_{i+1}$ is the cell in direction $\NW$ of $m_i$.
    Denote $\hole{}$ the hole containing that cell, and assume by contradiction that $\hole = \hole_m$.
    Consider a simple directed curve that starts and ends at the midpoints of the northern and southern side of $m_{i+1}$, respectively, and is entirely contained within $\hole$.
    Since holes are defined w.r.t. corner-adjacency, such a curve must exist, if $\hole = \hole_m$.
    Now connect the endpoints of that curve to a simple directed loop.
    Depending on the positioning of modules, that loop is either anti-clockwise or clockwise (see $\mathcal{C}_1$ and $\mathcal{C}_2$ in \cref{fig-appendix:active-critical-pair}).
    If the loop is anti-clockwise, then it encloses all successors of $m_{i+1}$, but not $m_k$.
    Otherwise, if the loop is clockwise, then it encloses no successor of $m_{i+1}$, but it encloses $m_k$.
    Both contradict that $m_k$ is a successor of $m_{i+1}$.
    Hence, it holds that $\hole{} \neq \hole_m$.
    Note that by placing a module at $m$, the target hole disconnects into precisely two distinct holes.
    Otherwise, $m_k$ would have only one side-adjacent neighbor, which implies that it is simply removable and contradicts that $m$ became \pActive{} before $m_k$.
    It follows that $\hole = \hole_b$.

    Second, consider the case where $c$ is not contained in $P$.
    Then $m_k$ is the first module in $P$ that is side-adjacent to $m$.
    This again follows since $m_k$ must have at least two neighbors in the round in which $m$ became \pActive{}, otherwise contradicting that $m$ becomes \pActive{} before $m_k$.
    Hence, $\hole_b$ must contain the target cell.
    
    We have shown that the \pActive{} module can determine whether $\hole_m$ or $\hole_b$ contains the target cell.
    If $\hole_m$ contains the target cell, then it sends a cleanup message to $m_k$ to delete all requests which followed the boundary of $\hole_b$.
    Note that if the \pActive{} module would become \pPassive{} instead in this case, then invariant (1) would be violated: by \cref{def:request-path}, $P$ is contained in the boundary of the target hole $\hole_m$, but since $\hole_b$ is distinct from $\hole_m$, they do not necessarily share the same boundary cells.
    Since $i < k$, the cleanup message must eventually arrive at $m_i$, after which the \pActive{} module is again side-adjacent to the last module of $P$.
    The cleanup message iteratively removes the last edge of $P$, such that the invariants carry over in each step.
    Otherwise, if $\hole_b$ contains the target cell, then the \pActive{} module becomes \pPassive{}.
    The boundary of $\hole_b$ contains all stored requests such that (1) holds.
    After becoming \pPassive{}, the module simulates receiving a request from $m_k$, i.e., it becomes $m_{k+1}$ in the next round, and (2) holds.
    Finally, (3) holds since modules delete all stored requests upon becoming \pActive{}.
    Once a \pActive{} module reaches the target cell and there are no cleanup messages currently being forwarded, the request path and passive segment are empty, and the invariants again hold trivially, as shown at the start of this proof.        

    This concludes that all invariants are maintained throughout execution.
\end{proof}

\begin{lemma}[Restatement of \cref{lem:request-loop-1}]
    \label{lem-appendix:request-loop-1}
    Let $P = (m_1, m_2, \ldots, m_k)$ be a well-formed request path with non-empty passive segment $S = (m_i, \ldots, m_k)$.
    Suppose the next anti-clockwise module along the target boundary after $m_k$ is some module $m_j$ with $i \leq j \leq k-2$.
    Then at least one of the modules $m_j, m_{j+1}, \ldots, m_k$ is conditionally removable.
\end{lemma}

\begin{proof}
    Since $j \leq k-2$, there is at least one module between $m_j$ and $m_k$ in $S$.
    Note that otherwise $m_j$ would be the only side-adjacent neighbor of $m_k$, as it is both its direct predecessor in $P$, and the next anti-clockwise module along the target boundary, in which case $m_k$ would be simply removable.
    Since $S$ is simple (as we assume $P$ to be well formed), it follows, that $Z \coloneqq (m_j, m_{j+1}, ..., m_k)$ is a simple cycle.

    First, assume that there is a critical pair $(c_1,c_2)$ with $c_1,c_2 \in Z$. 
    Since $Z$ is a cycle, each module in $Z$ has at least two side-adjacent neighbors within $Z$.
    By definition of a critical pair, both $c_1$ and $c_2$ have at most two side-adjacent neighbors.
    It follows directly that both $c_1$ and $c_2$ are conditionally removable.

    Second, assume the contrary.
    By \cref{def:request-path}, $Z \subseteq P$ is fully contained in the target boundary.
    Consider removing all modules side-adjacent to $m_j$ except for those contained in $Z$, i.e., except for $m_{j+1}$ and $m_k$.
    The remaining set of modules decomposes in at least two side-connected components, one of which contains $Z$ as its boundary (see \cref{subfig-appendix:request-loop-1,subfig-appendix:request-loop-2}).
    Let $\outline{}$ be the point set obtained by intersecting the union of all cells occupied by $Z$ with the target hole (including the removed neighbors of $m_j$).
    Since we assume that $Z$ contains no critical pair, $\outline{}$ is a simple polygon whose vertices are corners of the cells occupied by $Z$.

    There are two cases to distinguish.
    First, assume $\outline{}$ does not enclose the target hole (see \cref{subfig-appendix:request-loop-1}).
    Since we operate on the square grid, the exterior angle at each polygon vertex, defined as the angle between one side and the extension of the adjacent side, is $\pm \frac{\pi}{2}$.
    The sum of exterior angles of any simple polygon is $2\pi$.
    Thus, there are at least four vertices with exterior angle $\frac{\pi}{2}$, each corresponding to a cell sharing exactly two consecutive sides with the target hole.
    At most one of these cells can be occupied by $m_j$, so there are at least three conditionally removable modules in $Z$.

    Next, assume $\outline{}$ encloses the target hole (see \cref{subfig-appendix:request-loop-2}).
    The target cell is always adjacent to some non-passive module.
    Since $Z$ contains only \pPassive{} modules (as it lies entirely within the passive segment $S$), $\outline{}$ must enclose one additional cell besides the target cell.
    This implies that the width or height of $\outline{}$ is at least two cells.
    Thus, $\outline{}$ has two vertices with exterior angle $-\frac{\pi}{2}$ or a side of length at least two.
    Vertices with exterior angle $-\frac{\pi}{2}$ correspond to cells in $Z$ sharing two consecutive sides with the target hole (with $Z$ now lying on the outside of the polygon).
    At most one of these cells can be occupied by $m_j$, so the other is conditionally removable.
    If $\outline{}$ contains a side of length at least two, then there are consecutive $a, b \in Z$ both side-adjacent to the target hole. 
    Let $c$ and $d$ be the cells side-adjacent to $a$ and $b$, respectively, that are not in the target hole and not within $Z$ (see \cref{subfig-appendix:request-loop-2}).
    If $c$ is empty, then $a$ is conditionally removable.
    If $d$ is empty, then $b$ is conditionally removable.
    If neither is empty, both $a$ and $b$ are conditionally removable.

    In all cases, $Z$ contains a conditionally removable module, concluding the lemma.
\end{proof}

\begin{figure}[tbp]
        \centering
        \begin{subfigure}[b]{0.333\linewidth}
            \centering%
            \includegraphics[width=\linewidth,page=1]{request-loop}%
            \subcaption{}
            \label{subfig-appendix:request-loop-1}
        \end{subfigure}%
        \begin{subfigure}[b]{0.333\linewidth}
            \centering%
            \includegraphics[width=\linewidth,page=2]{request-loop}%
            \subcaption{}
            \label{subfig-appendix:request-loop-2}
        \end{subfigure}%
        \begin{subfigure}[b]{0.333\linewidth}
            \centering%
            \includegraphics[width=\linewidth,page=3]{request-loop}%
            \subcaption{}
            \label{subfig-appendix:request-loop-3}
        \end{subfigure}%
        \caption{((a, b) restated.) Illustration of the cases from \cref{lem-appendix:request-loop-1} and the terminology used in \cref{lem-appendix:request-loop-2}. (a) $\outline{}$ does not enclose the target hole; at least four vertices have exterior angle $\frac{\pi}{2}$. (b) $\outline{}$ encloses the target hole; at least two vertices have exterior angle $-\frac{\pi}{2}$ or there is a side of length at least two. (c) The C-shape in the common neighborhood of $m$ and its side-adjacent empty cell $e$.}
        \label{fig-appendix:request-loop}
    \end{figure}

\begin{lemma}
    \label{lem-appendix:extended-rhombus-weakly-convex}
    For any $n \geq 0$, the set of cells $\{v_0, \ldots v_{n}\} \cup N(v_n)$ is weakly convex, where $v_0, ..., v_n$ are the cells of a rhombus of size $n+1$.
\end{lemma}

\begin{proof}
    We first prove by induction that $\rhombus_n = \{v_0, \ldots, v_{n-1}\}$ is weakly convex for all $n \geq 0$.
    
    When $n = 0$, $\rhombus_0 = \emptyset$, which is trivially weak convex.
    Let $n > 0$, let $c_1, c_2 \in \rhombus_n$ be arbitrary, and assume that $\rhombus_{n-1}$ is weakly convex.
    We distinguish three cases:
    \emph{Case 1:}
    $c_1 = c_2 = v_{n-1}$.
    A single cell is trivially weakly convex.
    \emph{Case 2:}
    $c_1, c_2 \in \rhombus_{n-1}$.
    By the induction hypothesis, $\rhombus_{n-1}$ is weakly convex.
    Hence there is a shortest path from $c_1$ to $c_2$ fully within $\rhombus_{n-1} \subset \rhombus_{n}$.
    \emph{Case 3:}
    $c_1 = v_{n-1}, c_2 \in \rhombus_{n-1}$.
    In the square grid, all shortest paths are monotone paths, e.g., they go east and south but never west and north.
    Assume that $c_1$ is extremal within its rhombus layer, w.l.o.g. northernmost.
    Then all cells in $\rhombus_{n-1}$ lie in the southeast or southwest quadrant relative to $c_1$.
    We obtain a shortest path from $c_1$ to $c_2$ by taking one step south, and afterwards take a southeast or southwest monotone path within $\rhombus_{n-1}$, which exists by the induction hypothesis.
    Now assume that $c_1$ is not extremal within its rhombus layer.
    W.l.o.g., assume that it lies in the northwest quadrant of $\rhombus_n$.
    Then there is no cell in the northwest quadrant relative to $c_1$, and $c_1$ has neighbors in $\rhombus_{n-1}$ to the east and south.
    If $c_2$ lies in the southwest quadrant, take one step south and follow a southwest monotone path within $\rhombus_{n-1}$.
    If $c_2$ lies in the southeast quadrant, take one step south and follow a southeast monotone path within $\rhombus_{n-1}$.
    If $c_2$ lies in the northeast quadrant, take one step east and follow a northeast monotone path within $\rhombus_{n-1}$.
    In all cases, we take a single step from $c_1$ to some neighbor in $\rhombus_{n-1}$, and afterwards follow a shortest path within $\rhombus_{n-1}$ maintaining monotonicity.
    Hence, the result is a shortest path entirely within $\rhombus_{n}$, which concludes that $\rhombus_{n}$ is weakly convex for all $n \geq 0$.

    It remains to show that $R_n \coloneq \rhombus_{n+1} \cup N(v_n)$ is weakly convex.
    Let $c_1, c_2 \in R_n$ be arbitrary. 
    We distinguish three cases:
    \emph{Case 1:}
    $c_1, c_2 \in \{v_n\} \cup N(v_n)$.
    $\{v_n\} \cup N(v_n)$ is a $3 \times 3$ square of cells. 
    It is easy to see that a full square is weakly convex.
    \emph{Case 2:}
    $c_1, c_2 \in \rhombus_{n+1}$.
    Since $\rhombus_{n+1}$ is weakly convex, there is a shortest path from $c_1$ to $c_2$ within $\rhombus_{n+1} \subset R_n$.
    \emph{Case 3:}
    $c_1 \in N(v_n) \setminus \rhombus_{n+1}, c_2 \in \rhombus_{n+1}$.
    If $v_n$ is extremal within its rhombus layer, w.l.o.g. northernmost, then no cell in $\rhombus_{n}$ lies north of any cell in $N(v_n) \setminus \rhombus_n$.
    All cells in $N(v_n) \setminus \rhombus_{n+1}$ have a neighbor to the south.
    Hence, we can move from $c_1$ into $\rhombus_{n+1}$ by both southwest and southeast monotone paths, and afterwards take a path within $\rhombus_{n+1}$ that follows the same monotonicity.
    If $v_n$ is not extremal, w.l.o.g. it lies in the northeast quadrant of $\rhombus_{n+1}$, then there is no cell in the northeast quadrant relative to $c_1$ that is contained in $\rhombus_{n+1}$.
    We again distinguish the quadrant in which $c_2$ lies relative to $c_1$.
    If $c_2$ lies in the northwest quadrant, move west from $c_1$ into $\rhombus_{n+1}$ and follow a northwest-monotone shortest path to $c_2$.
    If $c_2$ lies in the southeast quadrant, move south from $c_1$ into $\rhombus_{n+1}$ and follow a southeast-monotone shortest path to $c_2$.
    If $c_2$ lies in the southwest quadrant, then: 
    Move west from $c_1$ as long as $c_2$ is not strictly south or until entering $\rhombus_{n+1}$, then move south until $c_2$ is not strictly west or until entering $\rhombus_{n+1}$, and once inside $\rhombus_{n+1}$, follow a southwest-monotone path to $c_2$.

    In all cases, we construct a shortest path from $c_1$ to $c_2$ that lies entirely within $R_n$.
    Hence, $R_n = \{v_0, ..., v_{n-1}\} \cup \{v_n\} \cup N(v_n)$ is weakly convex.    
\end{proof}

\begin{lemma}[Restatement of \cref{lem:request-loop-2}]
    \label{lem-appendix:request-loop-2}
    Let $P = (m_1, m_2, \ldots, m_k)$ be a well-formed request path with non-empty passive segment $S = (m_i, \ldots, m_k)$.
    Suppose the next anti-clockwise module along the target boundary after $m_k$ is not \pPassive{}, and that $m_k$ is not simply removable. 
    If no module in $S$ is conditionally removable, then the target hole has size one.
\end{lemma}
\newcommand{\fixed}{\mathcal{F}}

\begin{proof}
    The general structure of the proof is as follows.
    Let $D$ denote the set of all cells in $S$ that are side-adjacent to the target hole.
    Assuming that no module in $S$ is conditionally removable (see \cref{def:conditionally-removable}), we show the following three properties:
    (1) $D$ is non-empty, 
    (2) no critical pair $(c_1,c_2)$ exists with $c_1 \in D$ and $c_2 \in M$,
    and (3) for any consecutive $m_j, m_{j+1}$ in $S$, if $m_j \in D$, then $m_{j+1} \notin D$.
    We then consider the intersection of the neighborhood of any $m \in D$ with the neighborhood of some side-adjacent empty cell $e$.
    Using properties (1)--(3), we argue that modules within this intersection form a C-shape containing a module in $D$ that is distinct from $m$.
    The lemma then follows by considering the union of the C-shapes of two modules in $D$ together with property (2).
    We now proceed to prove each of the above claims in order.
    
    Let $m_l$ be the current position of the \pHead{} module within the rhombus $\rhombus{}$.
    Define $\fixed{}$ as the set of all \pHead{} and \pTail{} modules, i.e., the first $l+1$ cells of the rhombus.
    If there is a temporary \pHead{} module, i.e., we are in the special case handling described in Appendix~\ref{par:special-case}, then include $m_{l+1} \cup N(m_{l+1})$ in $\fixed{}$.
    We have shown in \cref{lem-appendix:extended-rhombus-weakly-convex} that $\fixed{}$ is a weakly convex set.
    Denote $m'$ the next anti-clockwise along the target boundary after $m_k$, and let $P'$ be a shortest path from $m'$ to $m_{i-1}$ that is fully contained within $\fixed{}$ (which exists since $\fixed{}$ is weakly convex).
    Note that $P'$ and $S$ are disjoint since $S$ contains only \pPassive{} modules, and $S$ is simple by the assumption that $P$ is well-formed.
    It follows that the concatenation $Z = S \circ P'$ is a simple cycle.

    Claim (1): $D$ is non-empty.
    Assume by contradiction that $S$ consists only of the module $m_k$.
    Let $m' \in \fixed{}$ denote the next anti-clockwise module along the target boundary after $m_k$.
    Since every request is initiated by some module in $\fixed{}$, it follows that $m_{k-1} \in \fixed{}$.
    If $m' = m_{k-1}$, then $m_k$ has only one side-adjacent neighbor, contradicting the assumption that it is not simply removable.
    Hence, assume $m' \neq m_{k-1}$.
    If $m'$ and $m_{k-1}$ are not corner-adjacent, then the unique shortest path from $m'$ to $m_{k-1}$ must include $m_k \notin \fixed{}$, contradicting the weak convexity of $\fixed{}$.
    Therefore, assume that $m'$ and $m_{k-1}$ are corner-adjacent.
    By weak convexity, there exists a shortest path from $m'$ to $m_{k-1}$ that includes a module corner-adjacent to $m_k$, again contradicting that $m_k$ is not simply removable.
    We conclude that $S$ contains at least two modules.
    Since $S$ forms a contiguous path along the target boundary, at least one of these modules must be side-adjacent to the target hole.
    It follows that $D$ is non-empty.

    Claim (2): No critical pair $(c_1,c_2)$ exists with $c_1 \in D$ and $c_2 \in M$.
    Each module $c_1 \in D$ has a side-adjacent empty cell (a cell in the target hole), and two side-adjacent neighbors in $Z$.
    If the fourth cell side-adjacent to $c_1$ is empty, then $c_1$ would be conditionally removable, contradicting our assumption that no module in $S$ is conditionally removable.
    It follows that $c_1$ has exactly three side-adjacent neighbors and therefore cannot be part of a critical pair, since any critical pair module has at most two side-adjacent neighbors.
    This concludes the proof of the claim.

    Claim (3): For any consecutive $m_j, m_{j+1}$ in $S$, if $m_j \in D$, then $m_{j+1} \notin D$.
    Assume by contradiction that both $m_j$ and $m_{j+1}$ are in $D$.
    As shown in the previous claim's proof, each of them must have exactly three side-adjacent neighbors: two of which are their predecessor and successor in $Z$.
    The side-adjacent neighbors of $m_{j}$ are $m_{j-1}, m_{j+1}$ and some module $a \in N(m_{j+1})$.
    Similarly, the side-adjacent neighbors of $m_{j+1}$ are $m_{j}$, its successor in $Z$ which is either $m_{j+2} \in S$ if it exists, or $m' \in \fixed{}$ otherwise, and some module $b \in N(m_{j})$.
    Together, these neighbors form a path from $a$ to $m_{j+1}$ via $b$ that is fully contained in the neighborhood $N(m_j)$.
    This implies that $m_j$ is conditionally removable, which contradicts our assumption that no module in $S$ is conditionally removable.
    Hence, no two consecutive modules in $S$ are both contained in $D$. 

    Let $m \in D$ be arbitrary (which exists by Claim (1)), and let $e$ be the empty cell in the target hole that is side-adjacent to $m$.
    By definition of $D$, such a cell $e$ exists, and it is unique, since $m$ is not conditionally removable, i.e., it has exactly three side-adjacent modules.
    We continue to show that the common neighborhood of $m$ and $e$ has a C-shape and contains a module from $D$ that is distinct from $m$.
    Consider the labeling of cells in $N(m) \cap N(e)$ as illustrated in \cref{subfig-appendix:request-loop-3}.
    In words: let $a$ and $b$ denote the successor and predecessor of $m$ in the cycle $Z$, respectively.
    Let $c$ and $d$ be the cells such that $(m, a, c, e)$ and $(m, b, d, e)$ each form a 4-cycle.

    First, consider the case where $a, b \notin \fixed{}$.
    In this case, there is a module at both $c$ and $d$; otherwise, it holds that $a \in D$ or $b \in D$, which would imply that two consecutive modules in $S$ belong to $D$, contradicting Claim (3).
    The unique shortest path from $c$ to $d$ goes through the empty cell $e$.
    If both $c$ and $d$ were in $\fixed{}$, this would contradict the weak convexity of $\fixed{}$.
    Hence, at least one of $c$ or $d$ must be outside of $\fixed{}$.
    Since both are side-adjacent to $e$, at least one of them must be contained in $D$.

    Second, consider the case where exactly one of $a, b$ is contained in $\fixed{}$.
    W.l.o.g., assume $a \notin \fixed{}$ and $b \in \fixed{}$.
    As in the previous case, there must be a module in cell $c$; otherwise $a \in D$, contradicting Claim (3).
    Any shortest path from $c$ to $b$ must contain $m$ or $e$.
    Since $m, e \notin \fixed{}$, weak convexity implies that $c \notin \fixed{}$.
    If there is a module in cell $d$, then there is nothing more to show.
    So assume by contradiction that cell $d$ is empty.
    Let $f_1, f_2$ and $f_3$ be the cells in $N(e)$ side-adjacent to $c, e$ and $d$, respectively (see \cref{subfig-appendix:request-loop-3}).
    Any shortest path from a cell $f_i$ to $b \in \fixed{}$ must contain one of the cells $c, e_s$ or $d$, all of which lie outside of $\fixed{}$.
    Thus, by weak convexity, none of the $f_i$ can be in $\fixed{}$.
    Moreover, none of the $f_i$ can be empty:
    If $f_1$ were empty, then $c \in D$ would be conditionally removable, contradicting our assumption.
    If $f_2$ were empty, then $f_1 \in D$ and $f_2 \in D$ would be consecutive in $S$, contradicting Claim (3).
    If $f_3$ were empty, then $f_2 \in D$ would be conditionally removable, again a contradiction.
    Finally, if $d$ were empty, then $f_2 \in D$ and $f_3 \in D$ would be consecutive in $S$, contradict Claim (3) once more.
    Hence, by contradiction, cell $d$ must be occupied.

    Third, the case $a, b \in \fixed{}$ cannot occur, as the shortest path from $a$ to $b$ contains $m \notin \fixed{}$.

    We have shown that for any $m \in D$, the modules in the intersection $N(m) \cap N(e)$ form a C-shape that includes a module $m^* \in D$ with $m^* \neq m$.  
    Together with Claim~(2), we now complete the proof of the lemma.

    Consider the union of the two intersections $N(m) \cap N(e)$ and $N(m^*) \cap N(e)$.  
    This union covers all cells in the neighborhood $N(e)$, except for a single cell that is corner-adjacent to $e$ (e.g., if $m^* = c$ in \cref{subfig-appendix:request-loop-2}, then the only cell not covered is $f_3$).
    By the weak convexity of $\fixed{}$, there can be no critical pair $(c_1, c_2)$ with $c_1, c_2 \in \fixed{}$.  
    Combined with Claim~(2), this implies that the target hole contains exactly one pseudo-hole.
    Consequently, the single remaining cell not contained in the union of the two C-shapes must be occupied as well. 

    We conclude that the target hole has size one.
\end{proof}

\begin{lemma}[Restatement of \cref{lem:fill-target}]
    \label{lem-appendix:fill-target}
    Consider a round in which a module initiates a request for some target cell, and assume at least one \pPassive{} module exists.
    Then, within $\O(n)$ rounds, either the target cell is filled or the target hole consists of just the target cell.
\end{lemma}

\begin{proof}
    Let $e$ be the target cell, and assume that the target hole has size larger one, as otherwise there is nothing to show.
    Consider the first round after the request for cell $e$ is initiated in which the request is \emph{not} forwarded to the next module along the target boundary.
    Let $P = (m_1, m_2, \ldots, m_k)$ be the request path and $S$ be its passive segment in that round.
    Let $\fixed{}$ be defined as in the proof of \cref{lem-appendix:request-loop-2}.
    By weak convexity of $\fixed{}$, no \pPassive{} module is fully enclosed by modules in $\fixed{}$.
    Moreover, since we maintain that the set of all modules is side-connected (since we only ever move simply removable and conditionally removable modules), and the request path traverses the target boundary starting in $\fixed{}$, it must eventually contain a \pPassive{} module, i.e., $S$ is non-empty.
    By \cref{lem-appendix:invariants}, $m_k$ is \pPassive{}.
    This implies, that in that round module $m_k$ does not forward the request because (1) it is simply removable, (2) the next anti-clockwise module along the target boundary after $m_k$ is contained in $S$, or (3) that module is contained in $\fixed{}$.
    In all three cases, some \pPassive{} module $m_i$ with $i \leq k$ becomes \pActive{}.
        
    Let $\phole_i$ be the pseudo-hole containing cell $m_i$ after removing its module, and denote $\phole_0$ be the pseudo-hole containing the target cell.
    Note that both simple and conditional removability requires $m_i$ to be side-adjacent to the target hole.

    First, consider the case where $\phole_i \neq \phole_0$. 
    Consider a shortest corner-path from $m_i$ to target cell $e$ entirely within the target hole $\hole$.
    Let $b_1$ be the last cell on that path within $\phole_i$, and $b_2$ the first cell on that path \emph{not} within $\phole_i$.
    Note that these two cells are uniquely defined, since for any adjacent pseudo-hole, $\phole_i$ contains exactly one bridge cell.
    Let $(c_1, c_2)$ be the critical pair adjacent to the bridge cells $b_1, b_2$.
    At least one of $c_1, c_2$ must be contained on the request path.
    Let $j$ be the maximum index w.r.t. $P$ for which $m_j = c_1$ or $m_j = c_2$.
    Since $m_i$ and $b_1$ are in the same pseudo-hole, by \cref{lem:move-in-phole}, the \pActive{} module can move from $m_i$ to $b_1$ remaining side-adjacent to the boundary of $\phole_i$.
    Moving and deleting all requests until it is side-adjacent to $m_j$ requires $\O(k-j)$ rounds.
    Let $\mathcal{S}$ be the number of square sides that are shared between an empty cell in the target hole and an occupied cell in its boundary.
    Note that while an occupied cell may belong to the boundaries of multiple holes, each such square side can be uniquely assigned to the boundary of a single hole.
    After the \pActive{} module becomes \pPassive{}, the target hole decomposes into precisely two holes, one of which contains the now empty cell $m_i$.
    Before $m_i$ became \pActive{}, all nodes on $P$ with index larger $j$ are contained in the boundary of that hole.
    Hence, it takes $\O(k-j)$ rounds to decrease $\mathcal{S}$ by $\Omega(k-j)$.
    Since $\mathcal{S} = \O(n)$ initially, it takes $\O(n)$ rounds until $\mathcal{S} \leq 4$, in which case either the target hole has size one, or $\phole_i = \phole_0$.

    Now consider the case where $\phole_i = \phole_0$.
    In this case, by \cref{lem:move-in-phole}, the \pActive{} module can directly move into the target cell.
    It takes $\O(2k - i)$ rounds until $m_i$ is \pActive{} and another $\O(i)$ rounds until it arrives at the target cell.
    Hence, it takes $\O(n)$ in total.
\end{proof}
}
\end{document}